\documentclass{sn-jnl} % Remove the [bst/sn-apacite] option
\usepackage{graphicx}%
\graphicspath{{imagenes/}} % Images folder
\usepackage{multirow}%
\usepackage{amsmath,amssymb,amsfonts}%
\usepackage{amsthm}%
\usepackage{mathrsfs}%
\usepackage[title]{appendix}%
\usepackage{xcolor}%
\usepackage{textcomp}%
\usepackage{manyfoot}%
\usepackage{booktabs}%
\usepackage{algorithm}%
\usepackage{algorithmicx}%
\usepackage{algpseudocode}%
\usepackage{listings}%
\usepackage{hyperref}
\usepackage[style=apa, backend=biber, natbib=true]{biblatex}
\emergencystretch=\maxdimen
\newtheorem{theorem}{Theorem}%  meant for continuous numbers
\newtheorem{proposition}[theorem]{Proposition}% 

\providecommand{\norm}[1]{\lVert#1\rVert}

\begin{document}

%USER DEFINED COLORS
\definecolor{c1}{HTML}{00468B}
\definecolor{c2}{HTML}{ED0000}
\definecolor{c3}{HTML}{42B540}
\definecolor{c4}{HTML}{0099B4}
\definecolor{c5}{HTML}{925E9F}
\definecolor{c6}{HTML}{FDAF91}

\title[Clustering multivariate functional data using the epigraph and hypograph indices: a case study on Madrid air quality.]{Clustering multivariate functional data with the epigraph and hypograph indices: a case study on Madrid air quality.}


\author*[1]{\fnm{Belén} \sur{Pulido}}\email{belen.pulido@uc3m.es}

\author[2,3]{\fnm{Alba M.} \sur{Franco-Pereira}}\email{albfranc@ucm.es}

\author[1,4]{\fnm{Rosa E.} \sur{Lillo}}\email{rosaelvira.lillo@uc3m.es}

\affil*[1]{\orgdiv{uc3m-Santander Big Data Institute (IBiDat)}, \orgname{Universidad Carlos III de Madrid}, \orgaddress{\street{Calle Madrid, 126}, \city{Getafe}, \postcode{28902}, \state{Madrid}, \country{Spain}}}

\affil[2]{\orgdiv{Department of Statistics and O.R.}, \orgname{Universidad Complutense de Madrid}, \orgaddress{\street{Plaza de Ciencias}, \city{Madrid}, \postcode{28040}, \state{Madrid}, \country{Spain}}}

\affil[3]{\orgdiv{Instituto de Matemática Interdisciplinar (IMI)}, \orgname{Universidad Complutense de Madrid}, \orgaddress{\street{Plaza de Ciencias}, \city{Madrid}, \postcode{28040}, \state{Madrid}, \country{Spain}}}

\affil[4]{\orgdiv{Department of Statistics}, \orgname{Universidad Carlos III de Madrid}, \orgaddress{\street{Calle Madrid, 126}, \city{Getafe}, \postcode{28902}, \state{Madrid}, \country{Spain}}}


\abstract{With the rapid growth of data generation, advancements in functional data analysis (FDA) have become essential, especially for approaches that handle multiple variables at the same time. This paper introduces a novel formulation of the epigraph and hypograph indices, along with their generalized expressions, specifically designed for multivariate functional data (MFD). These new definitions account for interrelationships between variables, enabling effective clustering of MFD based on the original data curves and their first two derivatives. The methodology developed here has been tested on simulated datasets, demonstrating strong performance compared to state-of-the-art methods. Its practical utility is further illustrated with two environmental datasets: the Canadian weather dataset and a 2023 air quality study in Madrid. These applications highlight the potential of the method as a great tool for analyzing complex environmental data, offering valuable insights for researchers and policymakers in climate and environmental research.

}




\maketitle

\section{Introduction}\label{sec1}

FDA has emerged as a powerful framework for analyzing data observed over continuous intervals, providing a more comprehensive understanding of underlying processes and capturing inherent variability. FDA represents data as functions rather than fixed points, offering new insights into various areas of knowledge such as medicine, economics, and environmental science. Univariate functional data refers to data where each function represents the evolution of a single variable over the continuum. A comprehensive overview of FDA can be found in \citet{ramsay} and \citet{ferraty}. More recent approaches for FDA can be found in \citet{horvath}, \citet{hsing}, and \citet{wang}. By modeling functions rather than discrete values, FDA enables to extract valuable information and to detect underlying patterns that may be obscured in traditional data analysis approaches. 

However, in many real-world scenarios, multiple variables evolve simultaneously over a continuum, leading to a multivariate functional dataset. 
The analysis of MFD offers a wealth of possibilities in numerous domains. For example, 
in environmental monitoring, multiple pollutants are often measured simultaneously across different locations over time. Incorporating the multidimensional nature of the data allows a deeper understanding of complex systems, facilitating more informed decision-making. This extension presents significant challenges, as it requires considering the interrelationships between different dimensions of the data and developing appropriate statistical tools for their efficient analysis. 
%These methodologies must manage high-dimensional data, and enable efficient clustering, classification, and regression analyses. 
Extending fundamental tools to the multivariate functional context, such as summary statistics, dimension reduction techniques, clustering, classification, and regression analyses, remains an active research area. Examples from environmental and health data illustrate this progress: \citet{di2015functional} and \citet{qian2024multivariate} for functional principal component analysis, \citet{carroll2021cross} focuses on data registration, \citet{acal2022functional} contributes to the analysis of variance for functional data, and \citet{matabuena2023estimating} contributes to regression techniques.
Together, these advancements enhance our understanding of complex environmental issues, enabling the identification of key challenges and guiding adaptive measures.

%In clustering problems, most methods for multivariate functional data are model-based (see \citet{schmutz2020}, \citet{anton2023model} and \citet{hael2023unveiling}. 
While clustering methodologies are well-established for multivariate data, they have grown considerably in the functional context to address key features such as infinite dimensions, irregular shapes, and complex dependencies. This has fueled interest in developing clustering techniques specifically for FDA. However, clustering techniques for functional data have primarily been developed for univariate cases, with limited extension to MFD. See \citet{traore2019}, \citet{wu2022functional} and \citet{pulido2023} for some examples in the one-dimensional context.

Despite this gap, there has been significant progress in applying multivariate techniques within the functional context. \citet{jacques2014S}, \citet{zhang2023review}, and \citet{gertheiss2024functional} offer comprehensive overviews of clustering methods for functional data, emphasizing that a greater body of work exists for one-dimensional cases than for MFD. Addressing the complexities of infinite-dimensional datasets requires innovative methodologies to advance this field, and recent research is actively contributing to this effort.

For instance, model-based strategies are explored in \citet{zeng2019simultaneous}, \citet{schmutz2020}, \citet{anton2023model},  and \citet{hael2024dynamic}. Additionally, \citet{ieva2013depth}, \citet{yamamoto2017}, and \citet{martino} present k-means-based approaches, while \citet{song2024multi} applies a multivariate clustering method to the functional principal components of multivariate data. Together, these studies highlight a growing interest in adapting clustering techniques to the unique challenges of MFD.

The primary aim of this work is to introduce a methodological advancement for clustering MFD. Building on the methodology proposed in \citet{pulido2023}, which applies a multivariate clustering approach to the epigraph and hypograph indices of data and their derivatives, we extend this approach to the multivariate functional context. This requires adapting these indices to the multivariate setting. \citet{ieva2020component} suggests an extension based on weighted averages of single indices, which can be applied to extend the outliergram by \citet{arribas} to the multivariate functional context. In this direction, we propose novel extensions that preserve previously overlooked relationships between the different components in the data. 
%The proposed methodology is evaluated on simulated datasets and applied to two environmental datasets, demonstrating its utility in identifying behavioral patterns and supporting data-driven decision-making.

This paper is organized as follows. Section \ref{multiv_ind} reviews the definition of the epigraph and hypograph indices in the univariate case, discusses the existing extension based on weighted averages of univariate indices for each component \citep{ieva2020component}, and introduces novel definitions of the epigraph and hypograph indices for MFD. At the end of Section \ref{multiv_ind}, the relationship between the multivariate and univariate indices is examined, along with various theoretical properties, with proofs provided in the Appendix.
In Section \ref{clust}, EHyClus methodology is presented for clustering MFD based on these indices and alternative clustering approaches are reviewed. Section \ref{simul} evaluates EHyClus by comparing its performance with benchmark methodologies on simulated datasets. Then, Section \ref{real} applies EHyClus to two real-world datasets: the Canadian Weather dataset and a dataset of $\text{NO}_2$ and $PM10$ concentrations in Madrid, where the optimal number of clusters is determined.
Finally, Section \ref{conc} concludes with remarks on the contributions and potential future research directions in multivariate functional data analysis.

\section{Multivariate epigraph and hypograph indices} \label{multiv_ind}

The one-dimensional definitions of the epigraph and hypograph indices were first introduced by \citet{franc2011} and have since been applied for various purposes in the literature. In this work, we adopt the one-dimensional definitions provided by \citet{martin2016} and propose novel extensions to the multivariate framework based on these definitions. Additionally, we consider the work of \citet{ieva2020component}, which introduces multivariate versions of the indices using weighted averages.

%For instance, \citet{arribas} exploit the relationship between the modified epigraph index and the modified band depth to develop an algorithm, known as the outliergram, used for shape outlier detection. \citet{martin2016} proposed the functional boxplot, a graphical technique that leverages the epigraph and hypograph indices to provide a novel definition of functional quartiles that exhibits greater robustness against shape outliers. \citet{franc2020} introduced a homogeneity test for functional data, based again on the combination of the epigraph and hypograph indices. Finally, \citet{pulido2023} showed their applicability in clustering functional data in one dimension. 

%\subsection{Notation and preliminaries} \label{prel}

One of the main purposes of this work is to broaden the application of the epigraph and the hypograph indices from a univariate to a multivariate context. Before proceeding, we first recall the definitions of the epigraph and hypograph indices for univariate functional data.

Let $C(\mathcal{I}, \mathbb{R})$ be the space of real continuous functions defined from a compact interval $\mathcal{I}$ to $\mathbb{R}$.  Consider a stochastic process $X \colon \mathcal{I} \longrightarrow \mathbb{R} $ with probability distribution $P_{X}$. The graph of a function $x$ in the space of continuous functions $C(\mathcal{I},\mathbb{R})$ is defined as $G(x) = \{(t,x(t)), \ \  \textit{for all} \ t \in \mathcal{I} \}.$ The epigraph (epi) and the hypograph (hyp) of a curve $x$ can then be introduced as follows:
$$epi(x)=\{(t,y) \in \mathcal{I} \times \mathbb{R} : y \geq x(t)\},$$
$$hyp(x)=\{(t,y) \in \mathcal{I} \times \mathbb{R} : y \leq x(t)\}.$$

Given a sample of curves $\{x_1(t),...,x_n(t)\}$, the epigraph and the hypograph  indices of a curve $x$ ($\mathrm{EI}_n(x)$ and $\mathrm{HI}_n(x)$ respectively) are defined as follows: 
%Note that the original definitions of these indices first appear in \citet{franc2011}, but several modifications have been performed in the literature for different methodologies. We will consider the definitions proposed in \citet{martin2016}, in the following way:
$$\mathrm{EI}_n(x)= 1-\frac{ \sum_{i=1}^n{I(\{G(x_i)\subseteq epi(x)\})}}{n}=1-\frac{\sum_{i=1}^nI\{x_i(t)\geq x(t), \  \textit{for all} \ t\in \mathcal{I}\}}{n},$$
$$\mathrm{HI}_n(x)=\frac{ \sum_{i=1}^n{I(\{G(x_i)\subseteq hyp(x)\})}}{n}=\frac{ \sum_{i=1}^n{I\{x_i(t)\leq x(t), \  \textit{for all} \ t \in \mathcal{I}\}}}{n}.$$

The epigraph index of a curve $x$ is defined as one minus the proportion of curves in the sample that are entirely contained in the epigraph of $x$, or equivalently, one minus the proportion of curves in the sample that are completely above $x$. In the same way, the hypograph index of $x$ represents the proportion of curves in the sample that are entirely included in the hypograph of $x$, or equivalently, the proportion of curves in the sample that are completely below $x$.

When there are many intersections between the curves in the sample, the previous definitions may become excessively restrictive, leading to values close to 1 and 0 for almost all the curves. Consequently, modified versions, denoted as $\mathrm{MEI}_n(x)$ for the epigraph index and $\mathrm{MHI}_n(x)$ for the hypograph index, are introduced to handle this issue:
\begin{equation} \label{mei_1d}
    \mathrm{MEI}_n(x) =1- \sum_{i=1}^n{ \frac{ \lambda ({t \in \mathcal{I}} : x_i(t) \geq x(t))}{n \lambda(\mathcal{I})}},
\end{equation}
\begin{equation}\label{mhi_1d}
    \mathrm{MHI}_n(x) = \sum_{i=1}^n{ \frac{ \lambda ({t \in \mathcal{I} : x_i(t) \leq x(t)})}{n \lambda(\mathcal{I})}},
\end{equation}
 where $\lambda$ stands for Lebesgue's measure on $\mathbb{R}$. These definitions allow for the interpretation of the indices as the proportion of time (when $\mathcal{I}$ is considered as a time interval) the curves in the sample are above or below $x$, respectively.

%As mentioned in the Introduction, it is important to have tools that facilitate the extension of statistical methodologies developed for the univariate case into the multivariate context. Since this constitutes a central objective in this paper, we have to introduce the following notation to deal with multivariate functional data. 
Let $C(\mathcal{I}, \mathbb{R}^p)$ be the space of real continuous functions defined from a compact interval $\mathcal{I}$ to $\mathbb{R}^p$.  Consider a stochastic process $\mathbf{X} \colon \mathcal{I} \longrightarrow \mathbb{R}^p $ with probability distribution $P_{\mathbf{X}}$. Let $\{\mathbf{x_1}(t), $ ... $\mathbf{x_n}(t)\}$ be a sample of curves from $P_{\mathbf{X}}$. Thus, 
\begin{align*}
    \mathbf{x_i} \colon \mathcal{I} &\to \mathbb{R}^p\\
        t          &\mapsto (x_{i1}(t), ..., x_{ip}(t))
\end{align*} where $i=1,\ldots,n.$
From now on, the multidimensional curves and the names of the multivariate indices are presented in bold font.

Numerous techniques for one-dimensional functional data rely on extremality indices, such as the outliergram by \citet{arribas}, the functional boxplot by \citet{martin2016}, and the homogeneity test by \citet{franc2020}. To expand the applicability of these methods into the multivariate context, a crucial first step is to generalize the underlying indices. The first extension of the epigraph and hypograph indices to the multivariate context is given by \citet{ieva2020component}, where they propose a definition of the MEI based on the extension of the band depth for MFD given in \citet{ieva2013depth}. This extension defines the multivariate MEI as a weighted average of the univariate counterparts. Given a set of functions $\textbf{x}_1(t),\ldots,\textbf{x}_n(t)$ the multivariate modified epigraph index of a multivariate curve $\mathbf{x}$  ($\mathbf{\rho MEI}_n$) is defined as the weighted average of the MEI values with respect to the sample curves $x_{1k}(t),\ldots, x_{nk}(t)$ for each component $k=1,\ldots,p$. To simplify the notation, $MEI_n(x_k)$ will denote the univariate MEI of the $k$-th component of the reference curve $\mathbf{x}_l$, with $1 \leq l \leq n$,  with respect to the univariate sample curves $x_{1k}(t),\ldots, x_{nk}(t)$.

\begin{equation} \label{rhomei}
    \mathbf{\rho MEI}_n(\mathbf{x})=\sum_{k=1}^p \rho_k \mathrm{MEI}_n(x_k),
\end{equation}
with $\rho_k > 0 \  \text{for all} \ k=1,\ldots,p,$ and $\sum_{k=1}^p \rho_k=1.$

The same approach can be followed to define the modified hypograph index, obtaining:
\begin{equation} \label{rhomhi}
    \mathbf{\rho MHI}_n(\mathbf{x})=\sum_{k=1}^p \rho_k \mathrm{MHI}_{n}(x_k).
\end{equation}

These definitions require a choice of the weights $\rho_k, \ k=1, \ldots, p,$ that, in general, is problem-driven, with no standard approach to calculate these weights. If there is no a priori knowledge about the dependence structure between the data components, these weights can be chosen uniformly, as $\rho_k=\frac{1}{p} \ \text{for all} \ k=1,\ldots,p.$ Alternative weight definitions have been suggested, relying on the variability of each component. \citet{ieva2020component} present a strategy to determine a data-driven set of weights $\{\rho_1,\ldots,\rho_p\}$, with $\rho_i=\frac{q_i}{\sum_{I=1}^p q_i}$, with $q_i=1/\lambda_i^{(1)}$ such that $\lambda_i^{(1)}$ is the maximum eigenvalue of the variance-covariance operator of the $i$-component, $\rho_i \geq 0, \  \text{for all} \ i=1,\ldots,p$ and $\sum_{i=1}^p \rho_i=1.$ 

%In this work, uniform and covariance-based weights will be considered. 
The multivariate modified epigraph and hypograph indices with uniform weights, referred to as \textbf{uMEI} and \textbf{uMHI}, are available in the R package \verb|roahd| \citep{ieva2019roahd}, and the definitions with covariance-based weights denoted as \textbf{cMEI} and \textbf{cMHI}, have been computed with our own implementation.

An evident limitation of previous definitions is their lack of consideration for the multivariate functional structure of curves. Our objective is to address this issue by extending the concepts of epigraph, hypograph, and their generalized versions to the multivariate functional context, incorporating the interdependencies among components of the curves. The proposed definitions compute the epigraph (or hypograph) index of a given curve, $\mathbf{x}$, as the proportion of curves with all their components fully above (or below) those of $\mathbf{x}$. These new definitions offer two key advantages: independence from data-driven weight assignments and inclusion of interdependencies between components, providing a more integrated view of MFD. 

The multivariate epigraph index of $\mathbf{x}$ ($\mathbf{EI}_n(\mathbf{x})$) with respect to a set of functions $\mathbf{x_1}(t), $ ... $\mathbf{x_n}(t)$ is defined as

\begin{equation}\label{ei1}
    \begin{split}
    \mathbf{EI}_n(\mathbf{x})& =  1-\frac{ \sum_{i=1}^n{I\{\bigcap_{k=1}^p\{G(x_{ik})\subseteq epi(x_k)\}\}}}{n}\\ & =  1-\frac{\sum_{i=1}^n{I\{\bigcap_{k=1}^p\{x_{ik}(t)\geq x_k(t), \  \text{for all} \ t \in \mathcal{I}\}}\}}{n} \\ & =
    1-\frac{\sum_{i=1}^n{\prod_{k=1}^p I\{x_{ik}(t)\geq x_k(t), \  \text{for all} \ t \in \mathcal{I}\}}}{n},
\end{split}
\end{equation}
 where $I\{A\}$ is 1 if $A$ true and 0 otherwise.
 
In the same way, the multivariate hypograph index of $\mathbf{x}$ ($\mathbf{HI}_n(\mathbf{x})$) with respect to a set of functions $\mathbf{x_1}(t), $ ... $\mathbf{x_n}(t)$ is defined as 

\begin{equation} \label{hi1}
    \begin{split}
        \mathbf{HI}_n(\mathbf{x}) & =  \frac{ \sum_{i=1}^n{I\{\bigcap_{k=1}^p\{G(x_{ik})\subseteq hyp(x_k)\}\}}}{n} \\& =\frac{\sum_{i=1}^n{I\{\bigcap_{k=1}^p\{x_{ik}(t)\leq x_k(t), \  \text{for all} \in \mathcal{I}\}}\}}{n} \\ & =
    \frac{\sum_{i=1}^n{\prod_{k=1}^p I\{x_{ik}(t)\leq x_k(t), \  \text{for all} \ t \in \mathcal{I}\}}}{n}.
    \end{split}
\end{equation}

Their population versions are given by: 

\begin{equation*}
    \mathbf{EI}(\mathbf{x},P_{\mathbf{X}})\equiv \mathbf{EI}(\mathbf{x}) =   1-P(\bigcap_{k=1}^p\{G(X_k) \subseteq epi(x_k)\})=  1-P(\bigcap_{k=1}^p\{X_k(t)\geq x_k(t),t\in \mathcal{I}\}),
\end{equation*} and,
\begin{equation*}
    \mathbf{HI}(\mathbf{x},P_{\mathbf{X}})\equiv \mathbf{HI}(\mathbf{x}) =   P(\bigcap_{k=1}^p\{G(X_k) \subseteq hyp(x_k)\})= P(\bigcap_{k=1}^p\{X_k(t)\leq x_k(t),t\in \mathcal{I}\}).
\end{equation*}

Analogous to the one-dimensional case, the definitions of the epigraph and the hypograph indices in multiple dimensions are highly restrictive. Consequently, it is necessary to introduce generalized versions of these two indices.

The multivariate generalized epigraph index of $\mathbf{x}$ ($\mathbf{MEI}_n(\mathbf{x})$) with respect to a set of functions $\mathbf{x_1}(t), $ ... $\mathbf{x_n}(t)$ is defined as 

\begin{equation} \label{mei1}
    \mathbf{MEI}_n(\mathbf{x})=  1-\frac{\lambda(\bigcap_{k=1}^p\{t \in \mathcal{I} : x_{ik}(t)\geq x_k(t)\})}{\lambda(\mathcal{I})}.
\end{equation}

In the same way, the generalized multivariate hypograph index of $\mathbf{x}$ ($\mathbf{MHI}_n(\mathbf{x})$) with respect to a set of functions $\mathbf{x_1}(t), $ ... $\mathbf{x_n}(t)$ is defined as 

\begin{equation} \label{mhi1}
    \mathbf{MHI}_n(\mathbf{x})= \frac{\lambda(\bigcap_{k=1}^p\{t \in \mathcal{I} : x_{ik}(t)\leq x_k(t)\})}{\lambda(\mathcal{I})}.
\end{equation}

If $\mathcal{I}$ is seen as a time interval, the multivariate generalized epigraph (hypograph) index of a given curve $\mathbf{x}$ can be understood as the proportion of time the curves in the sample have all their components totally above (below) $\mathbf{x}$. Note that these generalized definitions require that all the components are defined in the same interval $\mathcal{I}$.

The corresponding population versions of $\mathbf{MEI}_n(\mathbf{x})$ and $\mathbf{MHI}_n(\mathbf{x})$ are

$$\mathbf{MEI}(\mathbf{x},P_{\mathbf{X}})\equiv \mathrm{MEI}(\mathbf{x})=  1-\sum_{i=1}^n\frac{E(\lambda(\bigcap_{k=1}^p\{t \in \mathcal{I} : X_k(t)\geq x_k(t)\}))}{n\lambda(\mathcal{I})}, \ \text{and} $$

$$\mathbf{MHI}(\mathbf{x},P_{\mathbf{X}})\equiv \mathrm{MHI}(\mathbf{x})=  \sum_{i=1}^n\frac{E(\lambda(\bigcap_{k=1}^p\{t \in \mathcal{I} : X_k(t)\leq x_k(t)\}))}{n\lambda(\mathcal{I})}.$$

%In the one-dimensional case, the linear relation given by Equation~\eqref{rel1d} holds. 
Now, the relationship between the definitions of the epigraph and hypograph indices in the multivariate and the univariate cases are presented. The multivariate definitions of the indices, $\rho$\textbf{MEI} and $\rho$\textbf{MHI}, given by \citet{ieva2020component} (Equations~\eqref{rhomei} and~\eqref{rhomhi}) are obtained as a weighted average of the one-dimensional indices, thereby establishing a direct connection between these multivariate definitions and their one-dimensional counterparts. 
%However, the relationship between \textbf{MEI} and \textbf{MHI} (Equations~\eqref{mei1} and~\eqref{mhi1}) and MEI and MHI (Equations~\eqref{mei_1d} and~\eqref{mhi_1d}), and the relationship between \textbf{MEI} and \textbf{MHI} and $\rho$\textbf{MEI} and $\rho$\textbf{MHI}, warrant further examination.

A non-linear relationship can be established between \textbf{MEI} and \textbf{MHI}, which depends on the one-dimensional counterparts. This dependency also creates a connection with MEI and MHI and the weighted averages extensions $\rho$\textbf{MEI} and $\rho$\textbf{MHI}.

If we consider a multivariate functional dataset with dimensions $p>1$, the relationship between \textbf{MEI} and \textbf{MHI} depends on the values of the indices in all dimensions smaller than $p$. 
%Thus, a non-constant relationship between the two indices emerges. 
The following definitions and notation will be introduced to give an explicit formula of this relation: 
\begin{equation}\label{ap}
    A^p_{j_1,\ldots,j_r}=\sum_{i=1}^n \frac{\lambda(\bigcap_{k=1}^r \{ x_{ij_k} \geq x_{j_k} \}) }{n \lambda(I)},
\end{equation} and 
\begin{equation}\label{bp}
    B^p_{j_1,\ldots,j_r}=\sum_{i=1}^n \frac{\lambda(\bigcap_{k=1}^r \{ x_{ij_k} \leq x_{j_k} \} )}{n \lambda(I)},
\end{equation}
where $p$ is the number of dimensions of the initial dataset, and  $\{j_1,\ldots,j_r\} \subseteq \{1,\ldots,p\}$ denote the $r$ dimensions to be considered to define the index with dimension $r$. These $r$ dimensions form a permutation of size $r$ from the $p$ dimensions of the original dataset. In light of the preceding notation, the indices for a dataset consisting of $n$ functions in $p$ dimensions, are given as follows: 
\begin{equation}
    \mathbf{MEI}_n(\textbf{x})=1-A^p_{1,...,p},
    \label{mei}
\end{equation}
and 
\begin{equation}
    \mathbf{MHI}_n(\textbf{x})=B^p_{1,...,p}.
    \label{mhi}
\end{equation}

Now, the notation $\mathbf{MEI}^p_{n,j_1,...,j_r}$ and $\mathbf{MHI}^p_{n,j_1,...,j_r}$ will be considered to denote the epigraph/hypograph indices in dimension $r$ with $r \leq p$. The subset formed by $r$ of the $p$ dimensions conforming to the initial dataset, as mentioned before, will be denoted as $\{j_1,\ldots,j_r\}$.

In that way, 
\begin{equation}
    \mathbf{MEI}^p_{n,j_1,...,j_r}(\textbf{x})=1-A^p_{j_1,...,j_r},
    \label{mei_A}
\end{equation}
and 
\begin{equation}
    \mathbf{MHI}^p_{n,j_1,...,j_r}(\textbf{x})=B^p_{j_1,...,j_r},
    \label{mhi_B}
\end{equation}

If $r=p=1$, equations~\eqref{mei_1d} and~\eqref{mhi_1d} are particular cases of the equations~\eqref{mei_A} and~\eqref{mhi_B}, while if $r=p> 1$, equations~\eqref{mei} and \eqref{mhi} correspond to equations~\eqref{mei_A} and \eqref{mhi_B}, respectively. Thus, in order to simplify the notation,  $$\mathbf{MEI}^p_{n,j_1,...,j_p}(\textbf{x})=\mathbf{MEI}_n(\textbf{x}),$$ and $$\mathbf{MHI}^p_{n,j_1,...,j_p}(\textbf{x})=\mathbf{MHI}_n(\textbf{x}).$$

We are now poised to establish a relationship between the indices, which can be used for both one and multiple dimensional cases. \newline

\begin{theorem}\label{rel}
     The following relation between $\mathbf{MEI}_n$ and $\mathbf{MHI}_n$ holds for a dataset with $n$ curves in $p$ dimensions.  Let  $\mathbf{x_l}, \ 1 \leq l \leq n,$ be one of the sample curves, then the following relation holds, 
     \begin{align*}
         & \mathbf{MHI}_n(\mathbf{x_l})+(-1)^{p}\mathbf{MEI}_n(\mathbf{x_l})=  \\ & \sum_{r=1}^{p-1}\sum_{1\leq j_1<\ldots<j_r\leq p}^p (-1)^{r+p+1} \mathbf{MHI}^p_{n,j_1,\ldots,j_r}(\mathbf{x_l})+(-1)^{p+1}\frac{1}{n}+(-1)^{p+1}R_p.
     \end{align*}
     where $R_p = \sum_{k=1}^{2^p-1} \sum_{\substack{i=1 \\ i \neq j}}^n \frac{C}{n\lambda(I)},$ with $C \in \mathcal{C}_p$, where $\mathcal{C}_p$ is the set of the Lebesgue measure of all the possible intersections of $p$ elements of the type $\{x_{ij}> x_j\}$ or \\ $\{x_{ij}= x_j\}$, $j=1,\ldots,p$. 
     %where $R_p$ tends to 0 when $n$ tends to infinity and $\sum_{i=1}^n \lambda(\{t \in \mathcal{I} : x_i(t) = x_j(t)\}) = 0$.
\end{theorem} 

Note that, when evaluating this expression for $p=1$, we have that $$\text{MHI}_n(x)-\text{MEI}_n(x)=\frac{1}{n}+R_1.$$ In this case, $R_1=\sum_{\substack{i=1 \\ i \neq l}}^n \frac{\lambda\{x_{i}=x_l\}}{n\lambda(I)}$, which is 0 in case $\lambda\{x_{i}=x_l\}=0$, for $\ i \neq l$.

If the expression is now evaluated for $p=2$, then: $$\mathbf{MHI}_n(\mathbf{x})+\mathbf{MEI}_n(\mathbf{x})=\mathbf{MHI}^2_{n,1}(\mathbf{x})+\mathbf{MHI}^2_{n,2}(\mathbf{x})-\frac{1}{n}-R_2.$$ 

For $p=3$, the relationship will be given by: 
\begin{align*}
    & \mathbf{MHI}_n(\mathbf{x})-\mathbf{MEI}_n(\mathbf{x})= \mathbf{MHI}^3_{n,1,2}(\mathbf{x})+\mathbf{MHI}^3_{n,1,3}(\mathbf{x})+\mathbf{MHI}^3_{n,2,3}(\mathbf{x}) \\ & - \mathbf{MHI}^3_{n,1}(\mathbf{x})-\mathbf{MHI}^3_{n,2}(\mathbf{x})-\mathbf{MHI}^3_{n,3}(\mathbf{x})+\frac{1}{n}+R_3.
\end{align*}

 In order to facilitate comprehension of the general case, the proof when $p=3$ appears in Appendix~\ref{secA1}, along with the proof for the general case.

In summary, Theorem~\ref{rel} establishes a consistent relationship between \textbf{MEI} and \textbf{MHI} for MFD. Specifically, it demonstrates that this relationship remains constant in the one-dimensional case, where $R_1 = 0$. This is because, in both simulations and real data, it is rare for curves to overlap across intervals of positive Lebesgue measure.  

Note that one of the terms is $\sum_{i=1}^p \mathbf{MHI}_{n,i}^p$, which represents the sum of the generalized epigraph indices in one dimension, making it possible to establish a connection not only with the one dimensional indices (MEI and MHI) but also with the multivariate definitions introduced by \citet{ieva2020component} ($\rho$\textbf{MEI} and $\rho$\textbf{MHI}).

Now, we present several properties satisfied by \textbf{EI}, \textbf{HI}, \textbf{MEI} and \textbf{MHI} as given by Equations~\eqref{ei1},~\eqref{hi1},~\eqref{mei1} and~\eqref{mhi1}, respectively. They follow the line of \citet{lop2011}, \citet{ieva2013depth}, \citet{lop2014}, and \citet{franc2020}. The proofs of these results appear in Appendix~\ref{secA1}. \newline

\begin{proposition}
\label{p1}
The \textbf{EI} and \textbf{HI} with respect to a set of functions $\mathbf{x_1}(t), $ ... $\mathbf{x_n}(t)$ are invariant under the following transformations:
\begin{itemize}
    \item[a.] Let $\mathbf{T}(\mathbf{x})$ be the transformation function, defined as $\mathbf{T}(\mathbf{x}(t))=\mathbf{A}(t)\mathbf{x}(t)+\mathbf{b}(t),$ where $t \in \mathcal{I}$ and $\mathbf{A}(t) =\mbox{diag}(A_1(t), \dots, A_p(t))$ is a $p \times p$ matrix with $A_{j}(t)>0$ and $\mathbf{b}(t) \in C(\mathcal{I}, \mathbb{R}^p)$. Then, $$\mathbf{EI}(\mathbf{T}(\mathbf{x}))=\mathbf{EI}(\mathbf{x}), \ \mathrm{and,}$$
    $$\mathbf{HI}(\mathbf{T}(\mathbf{x}))=\mathbf{HI}(\mathbf{x}).$$
    
    \item[b.] Let $g$ be a one-to-one transformation of the interval $\mathcal{I}$ to $\mathcal{I}$. Then, $$\mathbf{EI}(\mathbf{x}(g))=\mathbf{EI}(\mathbf{x}), \ \mathrm{and,}$$
    $$\mathbf{HI}(\mathbf{x}(g))=\mathbf{HI}(\mathbf{x}).$$
\end{itemize}
\end{proposition}
 
The following proposition establishes similar properties as those mentioned in Proposition~\ref{p1}, but now for the generalized indices. \newline

\begin{proposition} \label{p1_2}
The \textbf{MEI} and \textbf{MHI} with respect to a set of functions $\mathbf{x_1}(t), $ ... $\mathbf{x_n}(t)$ are invariant under the following transformations:
\begin{itemize}
    \item[a.] Let $\mathbf{T}(\mathbf{x})$ be the transformation function defined as $\mathbf{T}(\mathbf{x}(t))=\mathbf{A}(t)\mathbf{x}(t)+\mathbf{b}(t),$ where $t \in \mathcal{I}$ and $\mathbf{A}(t) =\mbox{diag}(A_1(t), \dots, A_p(t))$ is a $p \times p$ matrix with $A_{j}(t)>0$ and $\mathbf{b}(t) \in C(\mathcal{I}, \mathbb{R}^p)$. Then, $$\mathbf{MEI}(\mathbf{T}(\mathbf{x}))=\mathbf{MEI}(\mathbf{x}), \ \mathrm{and,}$$
    $$\mathbf{MHI}(\mathbf{T}(\mathbf{x}))=\mathbf{MHI}(\mathbf{x}).$$
    
    \item[b.] Let $g$ be a one-to-one transformation of the interval $\mathcal{I}$ to $\mathcal{I}$. Then, $$\mathbf{MEI}(\mathbf{x}(g))=\mathbf{MEI}(\mathbf{x}), \ \mathrm{and,}$$
    $$\mathbf{MHI}(\mathbf{x}(g))=\mathbf{MHI}(\mathbf{x}).$$
\end{itemize}
\end{proposition}

The proposition below considers these indices as a measure of extremality. The objective is to demonstrate that these indices are suitable for ordering functions, as discussed in Section~\ref{order}. Specifically, \textbf{EI} arranges the sample of functions from bottom (\textbf{EI} equal to 0) to top (\textbf{EI} equal to 1). On the other hand, for \textbf{HI}, 1-\textbf{HI} is considered, where a value of 1 implies that there are no curves below it. Consequently, this index orders functions from top (1-\textbf{HI} equal to 0) to bottom (1-\textbf{HI} equal to 1). \newline

\begin{proposition} \label{p2}
The following results concerning the convergence of the maximum between $\mathbf{EI}$ and $1-\mathbf{HI}$ hold:

$$\underset{\min_{k=1,...,p}\norm{ x_k }_{\infty}\geq M}{\sup} \max \{\mathbf{EI}(\mathbf{x}, P_{\mathbf{X}}), 1-\mathbf{HI}(\mathbf{x}, P_{\mathbf{X}})\} \to 1, \ \text{when} \ M \to \infty, $$ and 
$$\underset{\min_{k=1,...,p}\norm{ x_k }_{\infty}\geq M}{\sup} \max \{\mathbf{EI}_n(\mathbf{x}), 1-\mathbf{HI}_n(\mathbf{x})\} \overset{a.s}{\to} 1, \ \text{when} \ M \to \infty$$
where $\norm{ x_k }_{\infty}$ is the supreme norm of the $kth$ component of $\mathbf{x}$.
\end{proposition}

From the strong law of large numbers, the strong consistency of $\textbf{EI}_n$ and $ \textbf{HI}_n$ to \textbf{EI} and \textbf{HI}, respectively follows inmediately. Proposition~\ref{p3} states this result. \newline

\begin{proposition} \label{p3} $\mathbf{EI}_n$ and $\mathbf{HI}_n$ are pointwise strongly consistent, meaning that

\begin{itemize}
    \item[a.] $\mathbf{EI}_n$ is strongly consistent. $$\mathbf{EI}_n(\mathbf{x}) \overset{a.s}{\to} \mathbf{EI}(\mathbf{x}, P_{\mathbf{X}}), \ as \ n \to \infty.$$
    \item[b.] $\mathbf{HI}_n$ is strongly consistent. $$\mathbf{HI}_n(\mathbf{x}) \overset{a.s}{\to} \mathbf{HI}(\mathbf{x}, P_{\mathbf{X}}) , \ as \ n \to \infty.$$
\end{itemize}
\end{proposition}

%\section{Ordering multivariate functional data}

%Finally, in this section, the orderings obtained with different indices for the same multivariate curves sample are compared. The orderings given by the following indices definitions are considered:

Finally, a comparison of the outputs/orderings given by MEI as given by Equation~\eqref{mei_1d}  for each dimension, \textbf{MEI} as given by Equation~\eqref{mei1}, and weight-based definition of the multivariate indices ($\rho$\textbf{MEI}) as given by Equations~\eqref{rhomei} is now given based on a toy example with six curves in two dimensions, represented in Fig.~\ref{curves_MEI_order}, which corresponds to Equation~\eqref{toy_example}. The figure on the left illustrates the first dimension of the curves, while the one on the right displays the second dimension. Each color corresponds to a distinct function, which facilitates a clear understanding of the association between the curves in the first dimension and those in the second dimension.

\begin{figure}[ht]
\centering
\includegraphics[width=\textwidth]{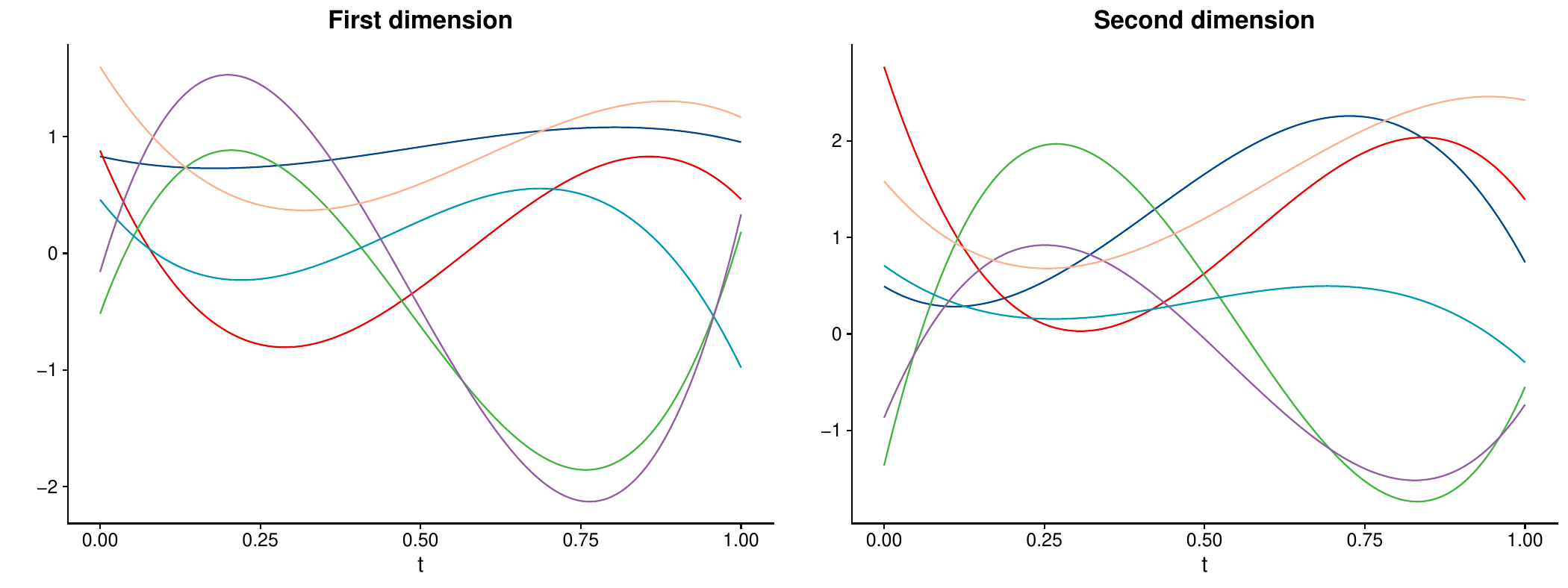}
\caption{Six distinct two-dimensional curves, each distinguished by a distinct color.  The left side showcases dimension 1, while the right side displays dimension 2. The colors when functions arranged from bottom to top when $t=0$, are green, purple, cyan, red, blue and orange in the first dimension, and green, purple, blue, cyan, orange and red in the second dimension.
\label{curves_MEI_order}}
\end{figure}

Table~\ref{order} provides a color assignment of the orderings obtained by the previously mentioned indices, ranging from bottom to top. These indices provide an ordering of data from top to bottom. When MEI is applied to each data component individually, it results in different orderings for each dimension, neglecting the interrelations among them. In contrast, \textbf{MEI} and $\rho$\textbf{MEI} offer unified orderings for the entire dataset. The main difference between them is that \textbf{MEI} takes interrelations between component into account, while $\rho$\textbf{MEI} is a weighted average of the individual indices.

%Table~\ref{order} provides a color assignment for curves based on various index values, ranging from bottom to top. Three different indices, resulting in three separate columns, are considered. The first column represents the order given by MEI for each dimension individually.  In this column, two values are presented: the first one corresponding to the curve achieving the given position in the first dimension and the second one corresponding to the curve in the second dimension. This two-values representation arises because MEI is a one-dimensional index. Conversely, the remaining two columns refer to \textbf{MEI} and \textbf{uMEI}, both of which consider the multivariate nature of the curves and, therefore, produce a single value in each case. 

\begin{table}[ht]

\centering
\begin{tabular}{cccc}
\toprule
 Ordering & (MEI1, MEI2) & \textbf{MEI} & \textbf{uMEI}\\
\midrule
1 & ({\color{c3}Green}, {\color{c5} Purple}) & {\color{c4} Cyan} & {\color{c5}Purple}\\
2 & ({\color{c2}Red}, {\color{c4} Cyan}) & {\color{c3} Green} & {\color{c4}Cyan}\\
3 & ({\color{c4}Cyan}, {\color{c3} Green}) & {\color{c5} Purple} & {\color{c3}Green}\\
4 & ({\color{c5}Purple}, {\color{c2} Red}) & {\color{c2} Red} & {\color{c2}Red}\\
5 & ({\color{c6}Orange}, {\color{c1}Blue}) & {\color{c1}Blue} & {\color{c1}Blue} \\
6 & ({\color{c1}Blue}, {\color{c6}Orange}) & {\color{c6}Orange} & {\color{c6}Orange}\\
\bottomrule
\end{tabular}
\caption{Color assignment for index values (1-6) indicating the ranking from the lowest value (top row) to the highest value (bottom row). The first column displays the MEI values for each component (ME1 for the first dimension and MEI2 for the second dimension), while the last two columns represent the multivariate indices \textbf{MEI} and \textbf{uMEI}. \label{order}}
\end{table}

When each dimension is considered independently, the resulting orderings differ, as shown in Table~\ref{order}. No curve holds the same position across both dimensions, underscoring the dissimilarity between them. In contrast, when all dimensions are considered together, the position assigned by the multivariate index may or may not coincide with the position of any univariate index in a particular dimension. This suggests that the extremeness of a curve depends on whether the interdependencies between dimensions are taken into account. Consequently, a curve may appear extreme in one dimension but not exhibit the same extremeness when evaluated multivariately.

Returning to the discussion of Fig.~\ref{curves_MEI_order}, the curve with the minimum \textbf{uMEI} is the purple one, which coincides to the minimum MEI in the second dimensions. In contrast, the curve with the minimum \textbf{MEI} is the cyan one. When these two curves are considered together, rather than independently, the cyan curve appears more extreme than the purple one. This underscores the importance of incorporating all dimensions of the curves into the index's definition.

In conclusion, this example highlights the significant impact of the chosen index on the resulting orderings.

\section{Clustering multivariate functional data} \label{clust}

The indices \textbf{MEI} and \textbf{MHI} naturally enable the adaptation of the methodology proposed in \citet{pulido2023} for clustering one-dimensional functional data to the multivariate context. In this section, we will provide an outline of this expansion, as well as an overview of various existing methods in the literature designed for clustering MFD. Subsequently, in the following sections, we will apply the proposed approach to various simulated and real datasets. Moreover, we will conduct a comparative analysis of the obtained results against those achieved by other established methodologies in the literature.

\subsection{EHyClus for multivariate functional data}\label{ehyclus}
The methodology proposed in \citet{pulido2023}, known as EHyClus, consists of four main steps:
\begin{enumerate}
    \item \textbf{Prepare the functional data.} Fit a cubic spline basis. Obtain the first and second derivatives of the data.
    \item \textbf{Apply indices to the data.} The epigraph, the hypograph and their generalized versions are applied to the data and their derivatives.
    \item \textbf{Apply multivariate clustering techniques.} Different multivariate clustering techniques are applied to different combinations of data and indices.
    \item \textbf{Obtain the best clustering partition of the data.} Apply different metrics to identify the best result.
\end{enumerate}

This approach transforms the original functional dataset into a multivariate one by applying the epigraph and hypograph indices in one dimension to the original curves, along with their first and second derivatives. Then, different multivariate clustering approaches are fitted to that dataset. Finally, a clustering partition is obtained as the combination of different indices and one clustering methodology. 

% \begin{figure}[ht]
% \centering
% \includegraphics[width=\textwidth]{ehyclus_steps.pdf}
% \caption{Strategy followed by EHyClus methodology. \label{EHyClus}}
% \end{figure}

In order to adapt EHyClus for the context of MFD, a modification is necessary in the second step, which involves applying indices to the data. This adjustment is needed to accommodate the multivariate dataset. There are several options for defining multivariate indices, including those introduced in this study (\textbf{MEI} and \textbf{MHI}), but also those proposed by \citet{ieva2020component} ($\rho$\textbf{MEI} and $\rho$\textbf{MHI}), with customizable weights or any other option.

In the one-dimensional case, the multivariate clustering techniques were applied to different combinations of the EI, HI and MEI of the curves and their first and second derivatives. Note that MHI was discarded because of the linear relation existing between MEI and MHI in practice. In the multivariate context, \textbf{EI} and \textbf{HI} are really restrictive and result, in almost all cases, in values so close to 1 and 0 respectively. This, added to the absence of a linear relation between \textbf{MEI} and \textbf{MHI} (see Section~\ref{multiv_ind}), leads to only consider \textbf{MEI} and \textbf{MHI}. A total of 15 different combinations of data, first and second derivatives with indices (Table~\ref{data}) were considered. In this table, the notation used can be expressed as (b).(c) where (b) represent the data combinations, being `\_' the original curves, `d' first derivatives and `d2' second derivatives, and (c) represents the indices that have been used. Once these 15 datasets are created, 12 different multivariate clustering techniques have been applied to each of them. These methods include different hierarchical clustering approaches with Euclidean distance, such as single linkage, complete linkage, average linkage and centroid linkage for calculating similarities between clusters, and Ward method \citep{hierarc}; k-means with Euclidean and Mahalanobis distances \citep{jain2010data}; kernel k-means (kkmeans) with Gaussian and polynomial kernels \citep{kkmenas}; spectral clustering (spc) \citep{spc} and support vector clustering (svc) with k-means and kernel k-means \citep{svc}. All these combinations result in 180 different cluster results denoted as (a).(b).(c) where (a) stands for the clustering method. See Table~\ref{comb}. 
To evaluate classification performance, three external validation strategies will be used: Purity, F-measure, and Rand Index (RI). These validation metrics are thoroughly explained in \citet{indexes} and \citet{rendon}. 

A key limitation of this methodology is its reliance on external validation metrics, which require ground truth data for calculation. To address this, an automated approach is proposed for selecting combinations of data and indices in real-world examples, using the percentage of distinct values per variable and the correlation between variables. 
The methodology is the following:
\begin{itemize}
    \item Calculate EI, HI, MEI and MHI (in one or multiple dimensions) on the data, first and second derivatives obtaining a 12 variables dataset.
    \item Discard those variables having less than 50\% of distinct values.
    \item Discard those variables with correlation greater than 75\%. 
    %To do so, we use the \verb|findCorrelation| function of the \verb|caret| R package.
\end{itemize}
The variables that have not been discarded are those used for EHyClus. When using this automated procedure, we will refer to it as auto-EHyClus. The only remaining decision in this approach concerns the choice of the clustering method. Based on the simulation study, k-means with Euclidean distance or spectral clustering are expected to perform particularly well. Fig.~\ref{auto-EHyClus-bp} presents the distribution of the difference between the RI obtained by auto-EHyClus, compared to the maximum RI among the 180 possible outcomes of EHyClus during 50 simulations of each data generation process (DGP). Each boxplot corresponds to one DGP among those in multiple dimensions considered in Section~\ref{simul}, and those in one dimension available in Section~4 in \citet{pulido2023}. The notation used in the boxplot to refer to each DGP is the one considered in these two works. A positive difference indicates that the index combination is not among the 180 outcomes, but improves the results. A negative difference means that the RI is worse. The fact that these differences are generally not positive suggests that the combination of data and indices, despite that not being all the possible combinations, are appropriately chosen. Moreover, examining the boxplots for different DGPs in both single and multiple dimensions, one can see that this difference has minimal impact, with the worst-case scenario showing a difference smaller than 0.25.

\begin{figure}[ht]
\centering
\includegraphics[width=\textwidth]{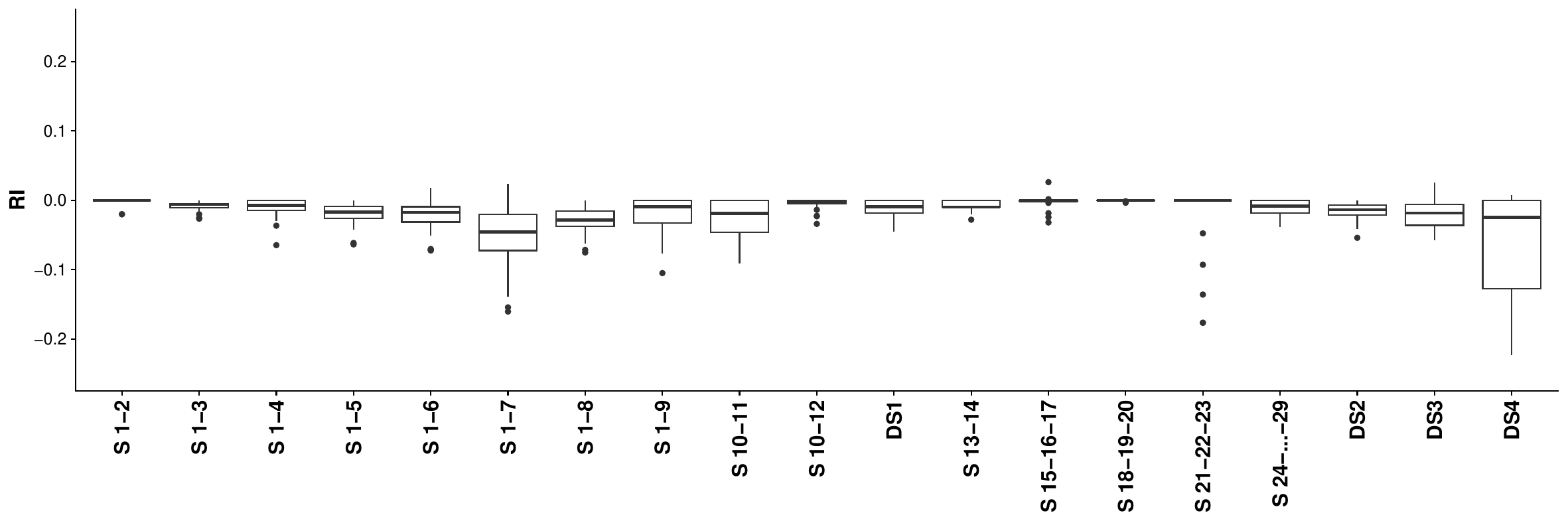}
\caption{Boxplot of the RI difference between auto-EHyClus and the best option among the 180 possibilities considered with EHyClus. \label{auto-EHyClus-bp}}
\end{figure}

\subsection{Clustering methods for multivariate functional data in the literature} \label{compar}

In this section, we present several existing approaches from the literature for clustering MFD. The outcomes of these approaches will be compared to the results of EHyClus. For benchmarking purposes, six distinct methods from the literature have been selected. Furthermore, to ascertain whether \textbf{MEI} and \textbf{MHI} offer more insights about the data compared to $\rho$\textbf{MEI} and $\rho$\textbf{MHI}, EHyClus as explained in Section~\ref{ehyclus}, has also been tested applying \textbf{uMEI} and \textbf{uMHI} (uniform weights) and \textbf{cMEI} and \textbf{cMHI} (weights based on the covariance matrices).

The first method for benchmarking is funclust algorithm, from \verb|Funclustering| R package, fully explained in \citet{jacques2014}. It is the first model-based approach for clustering MFD in the literature. This approach applies multivariate functional principal component analysis to the data, to posteriorly fit a parametric mixture model based on the assumption of normality of the principal component scores. One of the weaknesses of this strategy is that only a given proportion of principal components is modeled, leading to ignore some available information. This limitation is overcame by funHDDC algorithm, fully explained in \citet{schmutz2020}, and available in the \verb|funHDDC| R package. This methodology extends the latter by modeling all principal components with estimated variance different from zero. The next methodology is the FGRC method, described in \citet{yamamoto2017}. This strategy proposes a clustering method for MFD which combines a subspace separation technique with functional subspace clustering. It tries to avoid the clustering process to be affected by the variances among functions restricted to regions that are not related to true cluster structure. Then, kmeans-d1 and kmeans-d2 are two approaches described in \citet{ieva2013depth}. They are two different implementations of k-means, which basically differ in the distance considered between the multivariate curves. kmeans-d1 uses the norm in the Hilbert space $L^2(\mathcal{I},\mathbb{R}^p)$, while kmeans-d2 considers the norm in the Hilbert space $H^1(\mathcal{I},\mathbb{R}^p)$. Finally, the methodology proposed in \citet{martino} and available in the R package \verb|gmfd|, is also tested. This one is also based on k-means clustering, but in this case, a generalized Mahalanobis distance for functional data, $d_{\rho}$ where the value of $\rho$ has to be set in advance is employed. 

In this work, we will refer to these six techniques respectively as: funclust, funHDDC, FGRC, kmeans-d1, kmeans-d2 and gmfd-kmeans. Finally, EHyClus will refer to the methodology proposed in this work using \textbf{MEI} and \textbf{MHI}. EHyClus-mean will denote EHyClus with \textbf{uMEI} and \textbf{uMHI}, and EHyClus-cov will consider the use of \textbf{cMEI} and \textbf{cMHI}. Note that for the three options with EHyClus, the best result when considering external metrics is the one given in the tables in the next section. The small differences reflected in Fig.~\ref{auto-EHyClus-bp} make it possible to consider these top results in simulations.

\section{Simulation study} \label{simul}
This section encompasses different DGPs to illustrate the performance of the proposed methodology and to compare it with the existing approaches in the literature, explained in Section~\ref{compar}. These experiments serve to demonstrate the behavior and effectiveness of the proposed methodology in contrast to some other approaches available for clustering MFD.
Four different DGPs are simulated for this purpose, two (DS1 and DS2) with two groups, and another two (DS3 and DS4) with four groups. These DGPs are simulated 100 times and the mean results are presented. In the case of EHyClus, the best result based on the various metrics considered is the one displayed.

DS1 first appears in \citet{martino}, and is the extension of a one-dimensional example considered in the same work, which has also been employed in \citet{pulido2023} for clustering functional data in one dimension. It consists of two functional samples of size 50 defined in $[0,1]$, with continuous trajectories generated by independent stochastic processes in $L^2(\mathcal{I}^2)$. Each component of the curve is evaluated in 150 equidistant observations in the interval $[0,1]$. 

The 50 functions of the first sample are generated as follows: 
\begin{equation} \label{toy_example}
    \mathbf{X}_1(t)=\mathbf{E}_1(t)+ \sum_{k=1}^{100} \mathbf{Z}_k\sqrt{\rho_k}\theta_k(t),
\end{equation}
where $\mathbf{E}_1(t)=
\begin{pmatrix}
    t(1-t)\\
    4t^2(1-t)
\end{pmatrix}$
 is the mean function of this process, $\{\mathbf{Z}_k, k=1,...,100\}$  are independent bivariate normal random variables, with mean $\mathbf{\mu = 0}$ and covariance matrix 
 $\Sigma=
 \begin{pmatrix}
    1 & 0.5\\
    0.5 & 1
\end{pmatrix},$
 and $\{ \rho_k,k\geq 1 \}$ is a positive real numbers sequence defined as $$\rho_k = \left\{
	       \begin{array}{lll}
		 \frac{1}{k+1}      & if & k \in \{1,2,3\}, \\
		 \frac{1}{{(k+1)}^2} & if & k \geq 4,
	       \end{array}
	     \right. $$
in such a way that the values of $\rho_k$ are chosen to decrease faster when $k\geq 4$ in order to have most of the variance explained by the first three principal components. Finally, the sequence $\{\theta_k, k\geq 1\}$ is an orthonormal basis of $L^2(I)$ defined as$$\theta_k(t) = \left\{
	       \begin{array}{lllll}
		 I_{[0,1]}(t)    & \text{if} & k=1, &  \\
		 \sqrt{2}\sin{(k\pi t)}I_{[0,1]}(t) & \text{if} & k \geq 2,\\ & & k \ \text{even},\\
		 \sqrt{2}\cos{((k-1)\pi t)}I_{[0,1]}(t) & \text{if} & k \geq 3,\\ & & k \ \text{odd},
	       \end{array}
	     \right. $$
where $I_A(t)$ stands for the indicator function of set $A$.

The 50 functions of the second sample are generated by $$\mathbf{X}_2(t)=\mathbf{E}_2(t)+ \displaystyle \sum_{k=1}^{100} \mathbf{Z}_k\sqrt{\rho_k}\theta_k(t),$$ where $\mathbf{E}_2(t)=\mathbf{E}_1(t)+\mathbf{1} \displaystyle \sum_{k=4}^{100}\sqrt{\rho_k}\theta_k(t),$ is the mean function of this process, where $\mathbf{1}$ represents a vector of 1s.

The first step of EHyClus consists of smoothing the data with a cubic spline basis to remove noise and to be able to use its first and second derivatives. A sensitivity analysis regarding the best number of basis was carried out in \citet{pulido2023}, leading to the conclusion that a number of basis between 30 and 40 should be considered. In this section, as well as in that work, the number of basis is set to 35.

As shown in Fig.~\ref{martino21curves}, there is a significant overlap between the two groups in both dimensions, making it challenging to distinguish them visually.  However, upon examining the indices depicted in Fig.~\ref{martino21ind}, it becomes evident that the two groups can be discerned. This figure illustrates the utilization of \textbf{MEI} and \textbf{MHI} over the first derivatives. This representation has been executed in a two-dimensional format to enhance clarity of visualization, even though the best approach for DS1 includes four variables (\textbf{MEI} and \textbf{MHI} for both the first and second derivatives). 

\begin{figure}[ht]
\centering
\includegraphics[width=\textwidth]{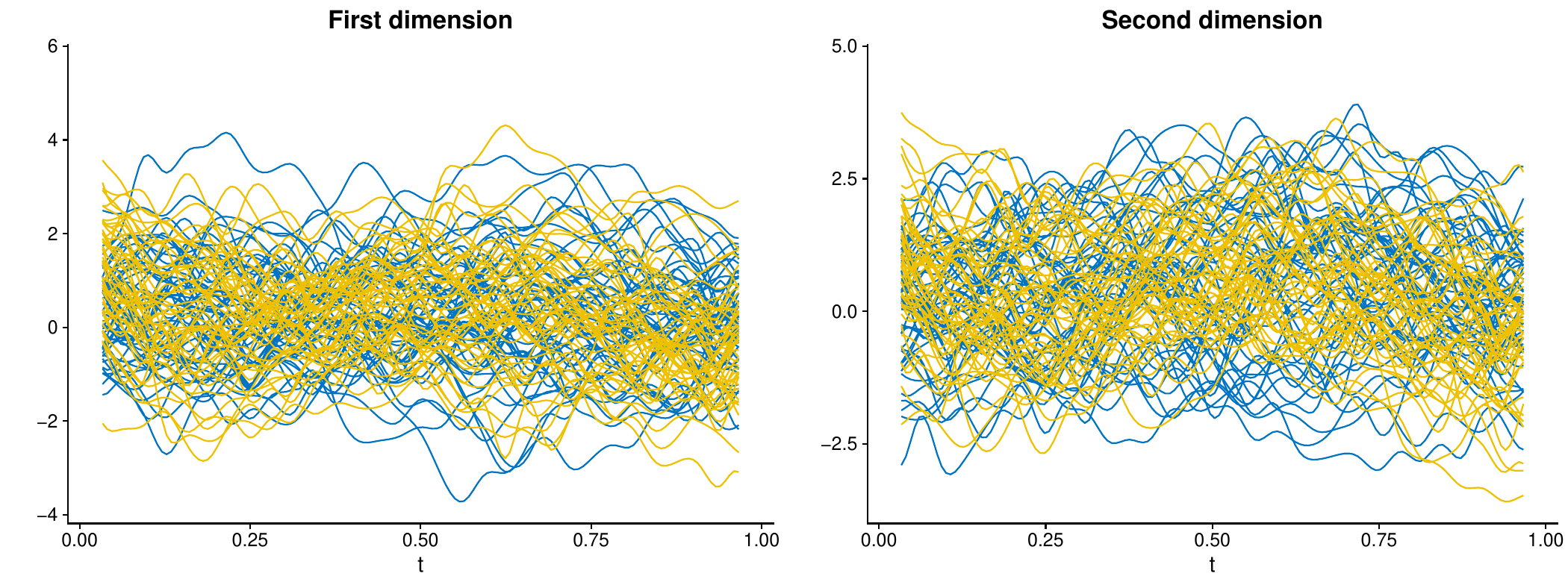}
\caption{DS1 data. Dimension 1 (left panel) and 2 (right panel).  \label{martino21curves}}
\end{figure}

\begin{figure}[ht]
\centering
\includegraphics[width=0.4\textwidth]{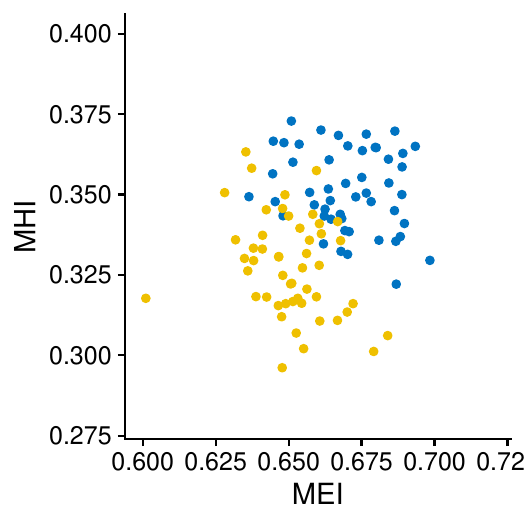}
\caption{Scatter plot of the \textbf{MEI} and the \textbf{MHI} of the first derivatives of DS1. \label{martino21ind}}
\end{figure}

The best approach of EHyClus, using the proposed indices, involves k-means with Euclidean distance, achieving a mean RI of 0.9698, as shown in Table~\ref{ds1comp}. Additionally, all the existing methods reviewed in Section~\ref{compar} are applied to DS1, and their mean results are also presented in the same table. It should be noted that this table reflects the best approach for each methodology when multiple options are available. Among these methods, kmeans-d2 achieves the highest value, 0.9009, which is approximately 0.07 units lower than EHyClus's best result. The next best method is Funclust, with a value of 0.8198, around 0.15 units lower than the proposed approach in this document. The remaining methods do not yield competitive results in terms of RI. This is clearly illustrated in Fig.~\ref{martino21boxplot}, which shows the distribution of RI for each method. While Funclust achieves a relatively high mean, it exhibits significant dispersion. Additionally, kmeans-d2 has a lower median than EHyClus, despite having a higher mean. These results suggest that EHyClus, with the proposed indices, is a competitive approach.

\begin{table}

\centering
\begin{tabular}{lcccc}
\toprule
  & Purity & Fmeasure & RI & Time \\
\midrule
EHyClus & 0.9846 & 0.9695 & 0.9698 & 0.00262 \\
EHyClus-mean & 0.7243 & 0.5986 & 0.6005 & 0.0106\\
EHyClus-cov & 0.7237 & 0.5977 & 0.5997 & 0.0104 \\
Funclust & 0.8563 & 0.8197 & 0.8198 & 1.3277 \\
funHDDC & 0.5810 & 0.5217 & 0.5157 & 3.6154 \\
FGCR & 0.5749 & 0.5070 & 0.5063 & 0.2275 \\
kmeans-d1 & 0.5635 & 0.4021 & 0.5034 & 0.1244\\
kmeans-d2 & 0.9964 & 0.8878 & 0.9009 & 0.1211 \\
gmfd-kmeans & 0.7400 & 0.6949 & 0.6678 & 3.3498 \\
\bottomrule
\end{tabular}
\caption{Mean values for DS1 of Purity, F-measure, Rand Index (RI) and execution time for EHyClus and all the competitors models on 100 simulations. \label{ds1comp}}
\end{table}

\begin{figure}[ht]
\centering
\includegraphics[width=0.8\textwidth]{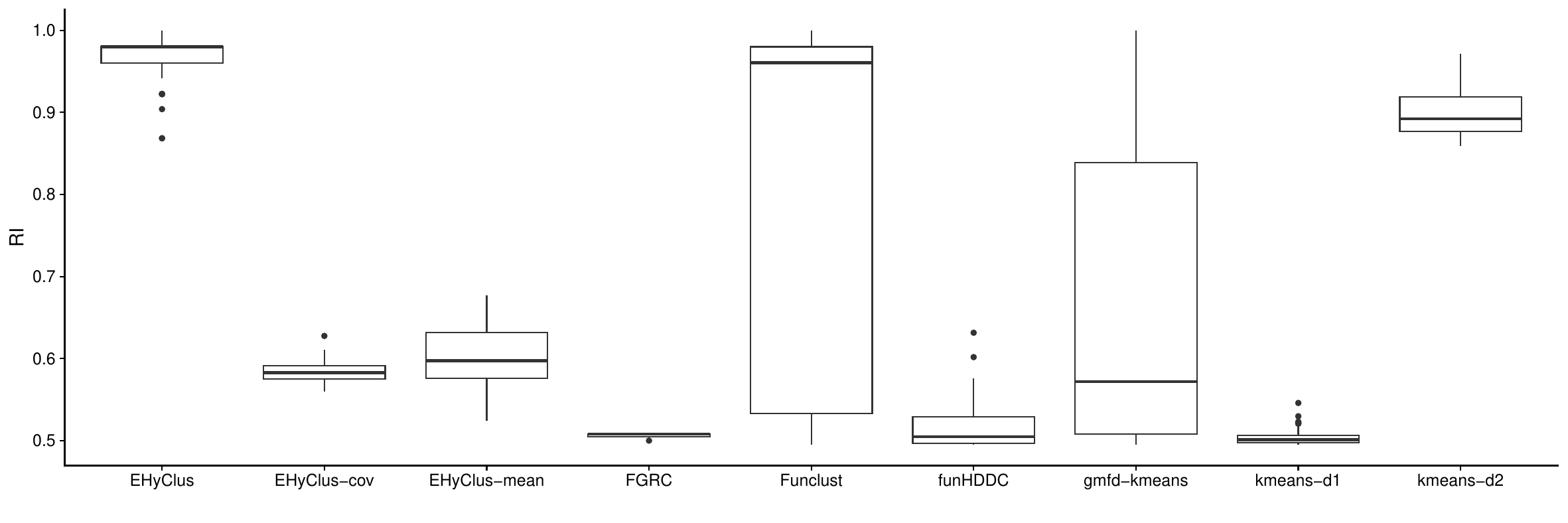}
\caption{Boxplot of the RI for DS1 over 100 simulation runs of EHyClus and its competitors. \label{martino21boxplot}}
\end{figure}

The second DGP (DS2) is based on a bivariate dataset with two groups appearing in \citet{jacques2014}. 
In this case, we consider 100 bivariate curves, with each component observed at 1001 equidistant points over the interval $[1,21]$. The first cluster consists of 50 functions generated by $X_{11}$ for the first dimension and $X_{12}$ for the second dimension. Similarly, the second cluster also comprises 50 functions generated by $X_{21}$ and $X_{22}$ for the first and second dimensions, respectively.
\begin{align*}
    X_{11}(t) = &-5+t/2+U_2h_3(t)+U_3h_2(t)+\sqrt{0.1}\epsilon(t), \\ 
    X_{12}(t) = & - 5+t/2+U_1h_1(t)+U_2h_2(t)+U_3h_3(t)+\sqrt{0.5}\epsilon(t) \\
     X_{21}(t) =& U_3h_2(t)+\sqrt{10}\epsilon(t), \\
    X_{22}(t) =&U_1h_1(t)+U_3h_3(t)+\sqrt{0.5}\epsilon(t),
\end{align*}
where $U_1 \sim \mathcal{U}(0.5,1/12)$, $U_2 \sim \mathcal{U}(0,1/12)$ and $U_3 \sim \mathcal{U}(0,2/3)$ are independent Gaussian variables and $\epsilon(t)$ represents a white noise independent of $U_i,$ $i=1,2,3,$ with unit variance. The functions $h_1$, $h_2$ and $h_3$ are defined as $h_1(t)=(6-|t-11|)_+$, $h_2(t)=(6-|t-7|)_+$ and $h_3(t)=(6-|t-15|)_+$, being $(\centerdot)_+$ the positive part. 

% The curves and its first derivatives for the two dimensions of the data are presented in Figures~\ref{DS2_firstdim} and~\ref{DS2_sndim}. The curves are highly mixed, but most pronounce overlap occurs in the derivatives. Figure~\ref{DS2_dermeans} represents the mean curves of the first derivatives over the two groups, showing the differences between them. This figure shows that the two groups present different behaviors, leading to be really well distinguishable when considering the indices, as represented in Figure~\ref{DS2dMEIMHI}.

% \begin{figure}[ht]
% \centering
% \includegraphics[width=\textwidth]{linear_curves1b.pdf}
% \caption{Dimension 1 of DS2 data. Original curves (left) and first derivatives (right). \label{DS2_firstdim}}
% \end{figure}

% \begin{figure}[ht]
% \centering
% \includegraphics[width=\textwidth]{linear_curves2.pdf}
% \caption{Dimension 2 of DS2 data. Original curves (left) and first derivatives (right).  \label{DS2_sndim}}
% \end{figure}

% \begin{figure}[ht]
% \centering
% \includegraphics[width=\textwidth]{linear_means.pdf}
% \caption{Mean curves of the four groups of DS2 derivatives. Dimension 1 (left) and 2 (right).\label{DS2_dermeans}}
% \end{figure}

% \begin{figure}[ht]
% \centering
% \includegraphics[width=0.5\textwidth]{linear_dMEIMHI.pdf}
% \caption{Scatter plot of the modified epigraph index (MEI) and the modified hypograph index (MHI) of the first derivatives of DS2. \label{DS2dMEIMHI}}
% \end{figure}

When applying EHyClus to DS2, more than 15 different combinations of data, indices, and clustering methods achieve perfect results across all three metrics: Purity, F-measure, and RI. All the combinations that lead to these perfect results include indices applied to the first derivatives, with some also incorporating indices from the second derivatives. Furthermore, a variety of clustering methods, including hierarchical options, k-means with Euclidean distance, and spectral clustering, are able to achieve these perfect outcomes. 

In contrast, when applying the seven methodologies used for comparison, only three can match the performance of EHyClus. As shown in Table~\ref{DS2comp}, only EHyClus-mean, EHyClus-cov, and funHDDC are competitive with EHyClus. However, funHDDC has significantly higher execution times compared to EHyClus under any index definition. Furthermore, when using alternative index definitions within EHyClus, more than 15 combinations again achieve perfect results, reaffirming that EHyClus consistently outperforms other methods for clustering MFD.

\begin{table}[ht]

\centering
\begin{tabular}{lcccc}
\toprule
  & Purity & Fmeasure & RI & Time \\
\midrule
EHyClus & 1.0000 & 1.0000 & 1.0000& 0.00739 \\
EHyClus-mean & 1.0000 & 1.0000 & 1.0000 & 0.0003\\
EHyClus-cov &1.0000 & 1.0000 & 1.0000 & 0.0003 \\
Funclust & 0.8386 & 0.8254 & 0.8062 & 4.8313 \\
funHDDC & 0.9897 & 0.9808 & 0.9808 & 5.8811 \\
FGRC & 0.8228 & 0.7839 & 0.7836 & 8.8738 \\
kmeans-d1 & 0.7775 & 0.7153 & 0.7165 & 0.0578\\
kmeans-d2 & 0.7618 & 0.6662 & 0.6671 & 0.0606 \\
gmfd-kmeans & 0.7211 & 0.6872 & 0.6649 & 53.7121 \\
\bottomrule
\end{tabular}
\caption{Mean values for DS2 of Purity, F-measure, Rand Index (RI) and execution time for EHyClus and all the competitors models on 100 simulations. \label{DS2comp}}
\end{table}

Finally, Fig.~\ref{DS2_boxplot} represents the RI distribution of each of the best approaches for each of the eight considered methodologies. EHyClus always obtains a RI equal to 1 for the three definitions of indices available in Section~\ref{multiv_ind}. The methodology called funHDDC also obtains almost all values equal to 1 in the 100 simulations. Nevertheless, it presents some outliers with a smaller RI. This implies that this approach does not obtain a mean RI equal to 1 in Table~\ref{DS2comp}. The remaining five methodologies obtain much more disperse results, with means much smaller than the other three approaches. Overall, EHyClus seems to be the best approach in this case.

\begin{figure}[ht]
\centering
\includegraphics[width=0.8\textwidth]{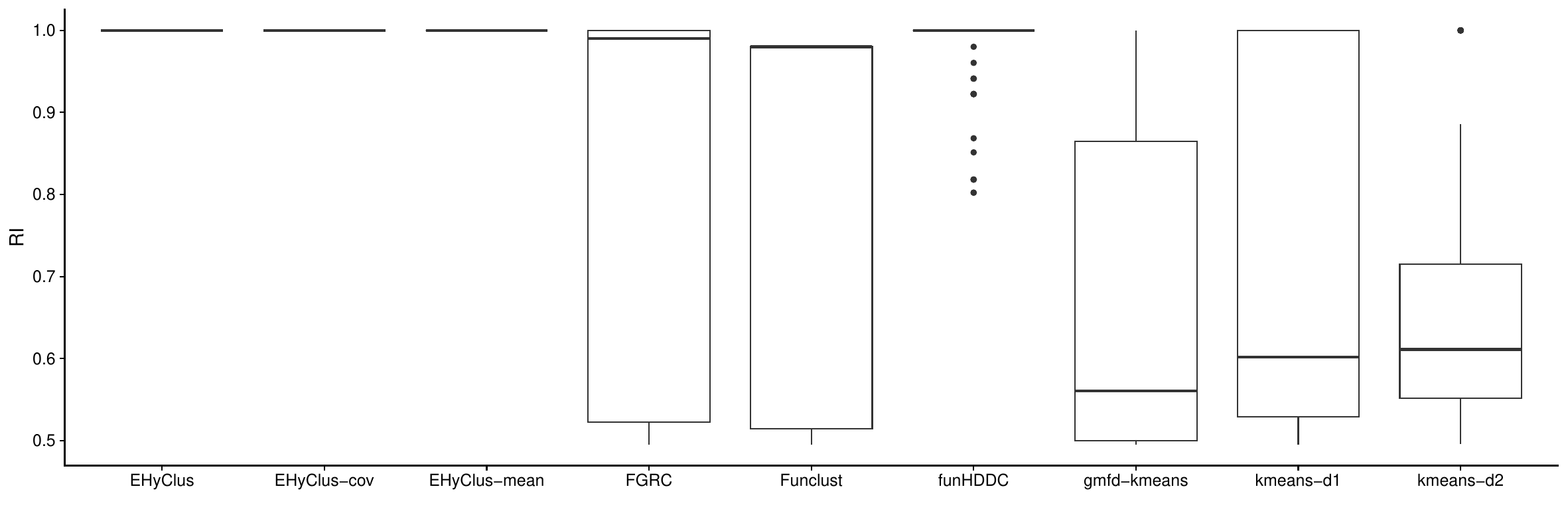}
\caption{Boxplot of the RI for DS2 over 100 simulation runs of EHyClus and its competitors. \label{DS2_boxplot}}
\end{figure}

The third DGP (DS3) has been previously considered by \citet{schmutz2020} to test their clustering algorithm. It is based on the data described in \citet{Bouveyron2015}, but changing the number of functions, the variance and adding a new dimension to the data. It consists of 1000 bivariate curves equally distributed in four different groups observed at 101 equidistant points of the interval $[1,21]$. Each cluster has this general form: $$X_1(t) = U + (A_1 - U)h_j(t) + \epsilon(t), \quad X_2(t) = U + (A_2 - U)h_k(t) + \epsilon(t)$$
% \begin{enumerate}
% \centering
%     \item[Cluster 1.] $X_1(t) =U+(1-U)h_1(t)+\epsilon(t),$ \\ $X_2(t) =U+(0.5-U)h_1(t)+\epsilon(t),$ 
%     \item[Cluster 2.] $X_1(t) =U+(1-U)h_2(t)+\epsilon(t),$ \\ $X_2(t) =U+(0.5-U)h_2(t)+\epsilon(t),$ 
%     \item[Cluster 3.] $X_1(t) =U+(0.5-U)h_1(t)+\epsilon(t),$ \\ $X_2(t) =U+(1-U)h_1(t)+\epsilon(t),$ 
%     \item[Cluster 4.] $X_1(t) =U+(0.5-U)h_2(t)+\epsilon(t),$ \\ $X_2(t) =U+(1-U)h_2(t)+\epsilon(t),$ 
% \end{enumerate} 
where $U \sim \mathcal{U}(0,0.1)$, $\epsilon(t)$ represents a white noise independent of $U$ with variance equal to $0.25$, and the functions $h_1$ and $h_2$ are defined as 
\begin{equation}\label{h2}
    h_1(t)=(a_1-|t-7|)_+
\quad\mathrm{and}\quad
    h_2(t)=(a_2-|t-15|)_+,
\end{equation}
with $a_1=a_2=6$.

The constants $A_1$ and $A_2$ are specific for each cluster, and $j$ and $k$ denote the index of the function h(t). In this way, DS3 is obtained as follows: 

\begin{enumerate}
\centering
\item[Cluster 1.]  $ A_1 = 1, \, A_2 = 0.5, \, j = 1, \, k = 1$ \\
\item[Cluster 2.]  $ A_1 = 1, \, A_2 = 0.5, \, j = 2, \, k = 2$ \\
\item[Cluster 3.]  $ A_1 = 0.5, \, A_2 = 1, \, j = 1, \, k = 1$ \\
\item[Cluster 4.]  $ A_1 = 0.5, \, A_2 = 1, \, j = 2, \, k = 2$
\end{enumerate}

The curves and first derivatives, when applying 35 cubic splines, are represented in Figures~\ref{DS3_firstdim} and~\ref{DS3_sndim}.

\begin{figure}[ht]
\centering
\includegraphics[width=\textwidth]{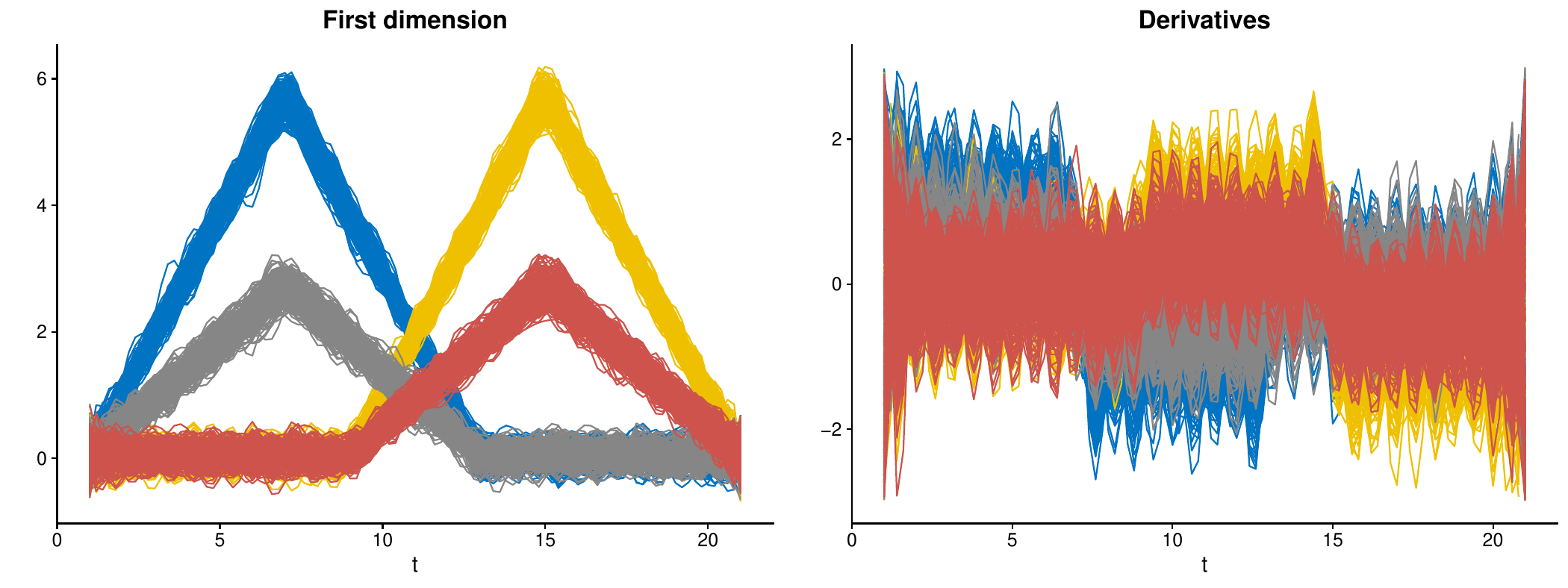}
\caption{Dimension 1 of DS3 data. Original curves (left) and first derivatives (right). \label{DS3_firstdim}}
\end{figure}

\begin{figure}[ht]
\centering
\includegraphics[width=\textwidth]{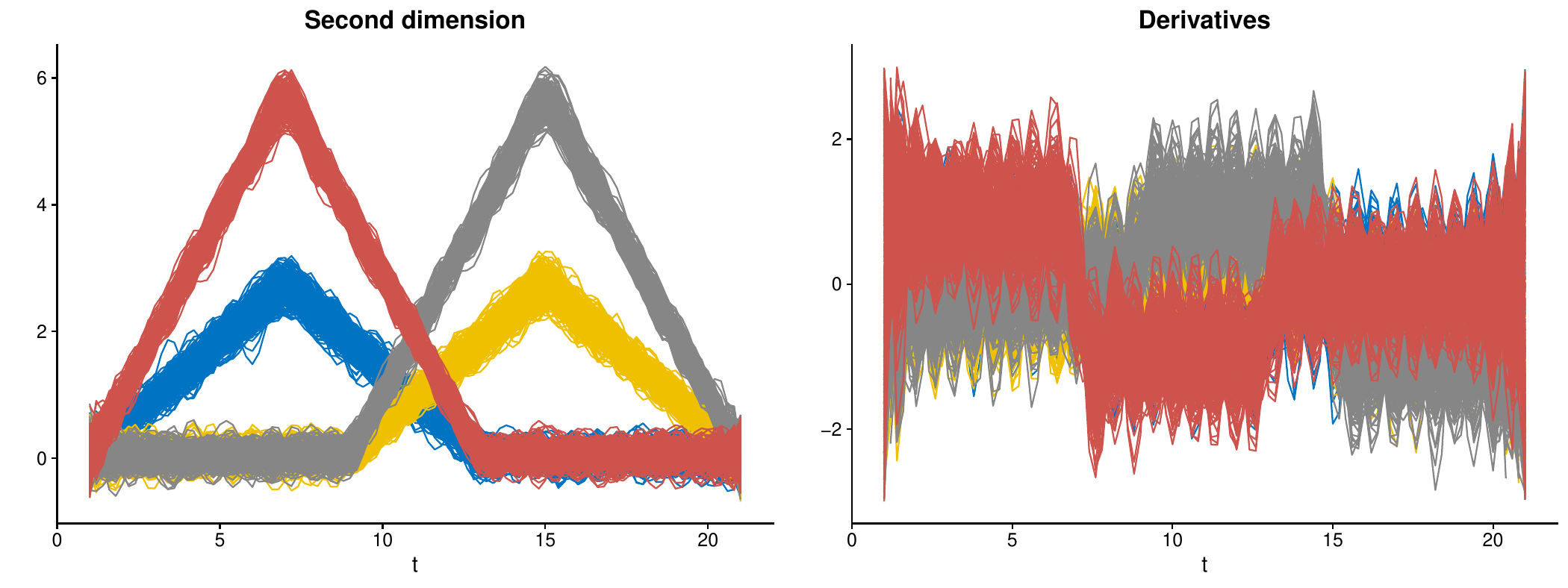}
\caption{Dimension 2 of DS3 data. Original curves (left) and first derivatives (right).  \label{DS3_sndim}}
\end{figure}

%The results derived from the application of EHyClus and eight alternative methodologies to this dataset are presented in Table~\ref{DS3ehycluscomp}. 
EHyClus produces the most favorable result when operating on the derivatives, and not on the original curves, obtaining the optimal combination employing k-means clustering on \textbf{MEI} and \textbf{MHI} derived from the first derivatives of the data. This finding is unexpected, as an examination of the curves displayed in Figures~\ref{DS3_firstdim} and~\ref{DS3_sndim} reveals that the groups are more distinguishable in the original curves compared to the derivatives. 
This phenomenon may be attributed to the fact that, owing to the shape of the derivatives, the disparity in the number of curves situated below and above a particular one provides a more effective discriminative capacity than in the case of the original curves. It is noteworthy that the methodology introduced by \citet{schmutz2020}, funHDDC, achieves exceptionally high results, far from those obtained by all the other approaches. See Table~\ref{DS3ehycluscomp}. When comparing all alternatives to funHDDC (0.99844 mean RI), EHyClus is the next best approach (0.88277 mean RI), far from Funclust (0.83591 mean RI), which is the following best value. Note that the execution time of funHDDC is really high compared to all the other approaches.%, being more than one thousand times slower than EHyClus and almost 5 times slower than Funclust. 

\begin{table}

\centering
\begin{tabular}{lcccc}
\toprule
  & Purity & Fmeasure & RI & Time \\
\midrule
%EHyClus & 0.81308 & 0.76515 & 0.88277 & 0.01423 \\
EHyClus & 0.81308 & 0.76515 & 0.88277 & 0.01423 \\ 
RHyClus-cov & 0.40440 & 0.29890 & 0.64860 & 0.00970 \\
EHyClus-mean & 0.39877 & 0.31863 & 0.65815 & 0.05580\\
Funclust & 0.70936 & 0.7674 & 0.83591 & 3.42640 \\
funHDDC & 0.99750 & 0.99737 & 0.99844 & 15.9790 \\
FGRC & 0.65151 & 0.63691 & 0.81316 & 6.38450 \\
kmeans-d1 & 0.49026 & 0.46950 & 0.73467 & 0.48181\\
kmeans-d2 & 0.35296 &  0.32660 & 0.66260 & 0.44237 \\
gmfd-kmeans & 0.48420 & 0.59104 & 0.66568 & 0.46780 \\
\bottomrule
\end{tabular}
\caption{Mean values for DS3 of Purity, F-measure, Rand Index (RI) and execution time for EHyClus and all the competitors models on 100 simulations. \label{DS3ehycluscomp}}
\end{table}

The preceding analysis conducted on DS3 was carried out using the dataset selected in \citet{schmutz2020}. The funHDDC methodology proposed in that research yielded remarkably high outcomes. To gain further understanding of how EHyClus operates with four groups, and to elucidate how funHDDC works in different scenarios, we believe it would be interesting to modify certain parameters in the formulation of DS3 and observe the resulting effects. Consequently, a new dataset, referred to as DS4, has been generated in the same way as DS3 with some changes in the considered parameters. In this case, $a_1=3$, $a_2=6$ and the four clusters are generated as follows:

\begin{align*}
    \text{Cluster 1. } & A_1 = 1.5, \ A_2 = 1, \ j = 1, \ k = 1 \\
    \text{Cluster 2. } & A_1 = 1, \ A_2 = 0.5, \ j = 2, \ k = 2 \\
    \text{Cluster 3. } & A_1 = 1, \ A_2 = 1, \ j = 1, \ k = 2 \\
    \text{Cluster 4. } & A_1 = 0.5, \ A_2 = 0.5, \ j = 2, \ k = 1
\end{align*}

Figures~\ref{DS4_firstdim} and~\ref{DS4_sndim} represent the original curves and first derivatives for the two dimensions of the data, that can be compared to those of DS3.

\begin{figure}[ht]
\centering
\includegraphics[width=\textwidth]{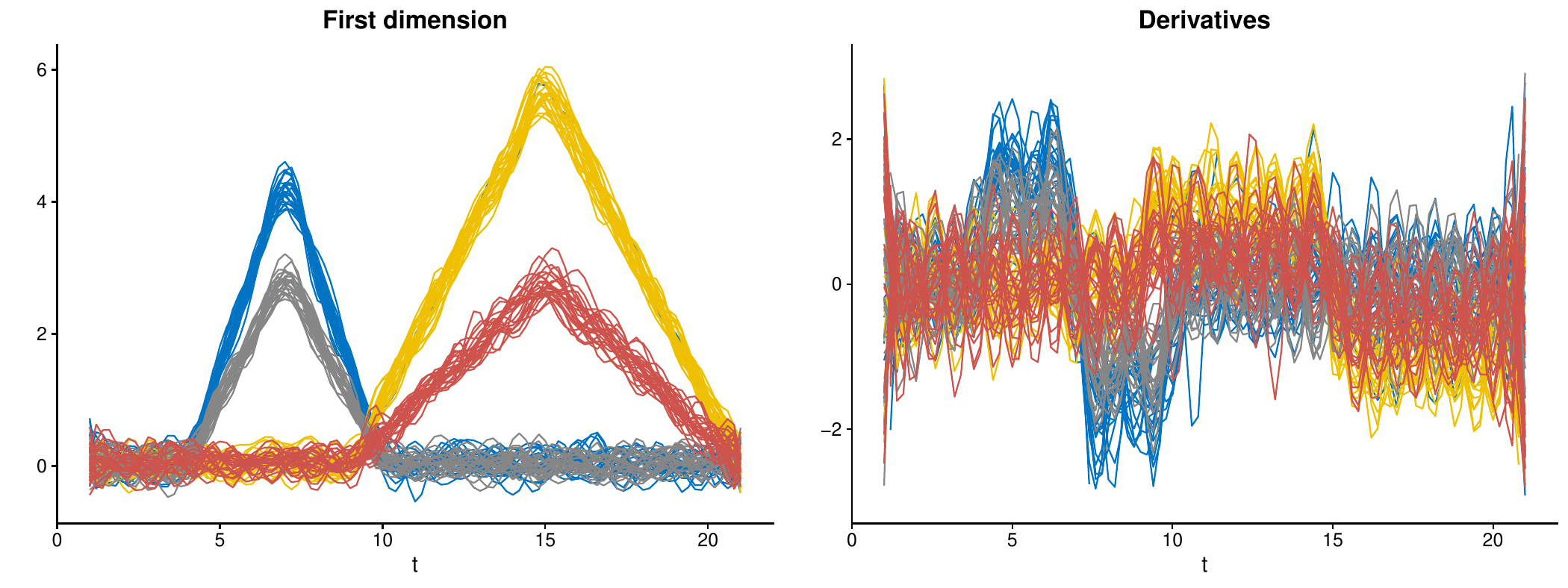}
\caption{Dimension 1 of DS4 data. Original curves (left panel) and first derivatives (right panel). \label{DS4_firstdim}}
\end{figure}

\begin{figure}[ht]
\centering
\includegraphics[width=\textwidth]{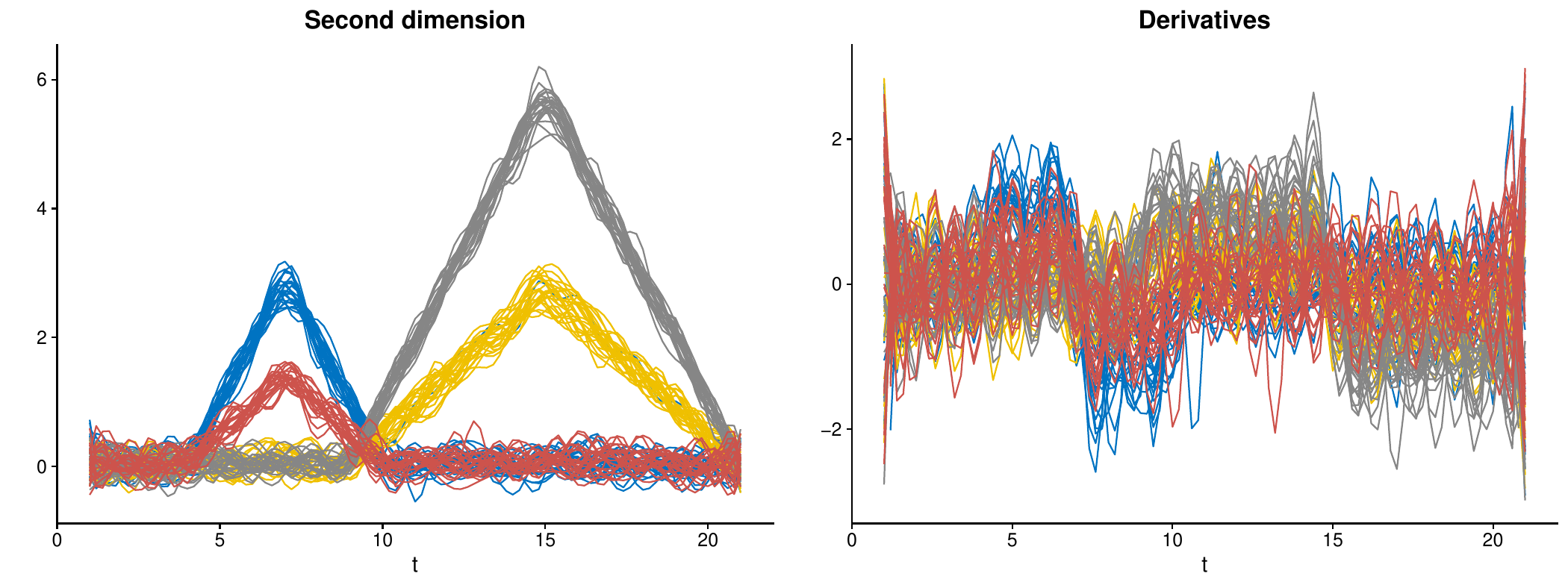}
\caption{Dimension 2 of DS4 data. Original curves (left panel) and first derivatives (right panel).  \label{DS4_sndim}}
\end{figure}

EHyClus obtains its best RI when applying k-means with Euclidean distance to \textbf{MEI} and \textbf{MHI} over data, first and second derivatives of DS4 data (0.9703 mean RI). Table~\ref{DS4comp} shows that, now, EHyClus outperforms funHDDC (0.8886 mean RI). EHyClus-mean and FGRC also are two approaches obtaining similar values as those achieved with funHDDC. 

\begin{table}

\centering
\begin{tabular}{lcccc}
\toprule
  & Purity & Fmeasure & RI & Time \\
\midrule
EHyClus & 0.9684 & 0.9392 & 0.9703 & 0.0080 \\
EHyClus-mean & 0.7382 & 0.6232 & 0.8142 & 0.0098\\
EHyClus-cov & 0.5813 & 0.4528 & 0.7186 & 0.0235 \\
Funclust & 0.6682 & 0.6986 & 0.7908 & 0.0914 \\
funHDDC & 0.8376 &  0.8187 & 0.8886 & 1.4595 \\
FGRC & 0.6772 & 0.6339 & 0.8163 & 0.11721 \\
kmeans-d1 & 0.3962 & 0.3136 & 0.6614 & 0.0261\\
kmeans-d2 & 0.4170 & 0.3350 & 0.6685 & 0.03113 \\
gmfd-kmeans & 0.7699 & 0.7389 & 0.8457 & 3.8594\\
\bottomrule
\end{tabular}
\caption{Mean values for DS4 of Purity, F-measure, Rand Index (RI) and execution time for all the competitors models on 100 simulations. \label{DS4comp}}
\end{table}

The distribution of RI for the nine methods, shown in Fig.~\ref{DS4_boxplot}, demonstrates that EHyClus is the best option, while funHDDC has the least dispersion. However, despite its higher dispersion, EHyClus produces much more accurate results than funHDDC in this case.

\begin{figure}[ht]
\centering
\includegraphics[width=0.8\textwidth]{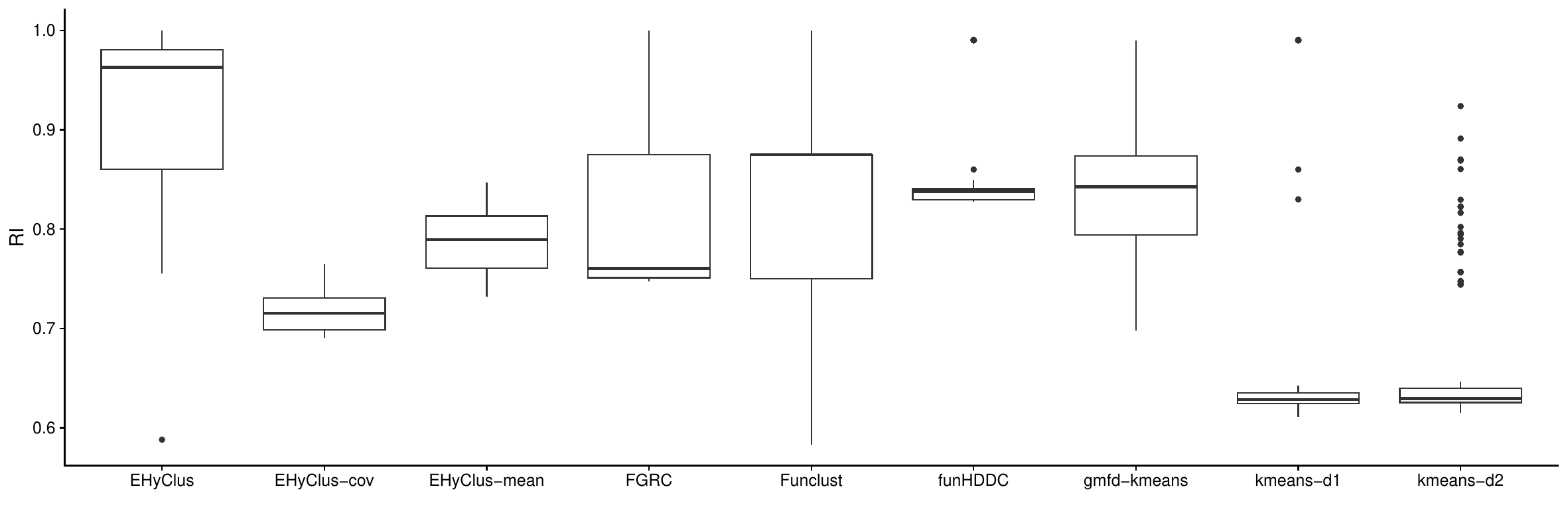}
\caption{Boxplot of the RI for DS4 over 100 simulation runs of EHyClus and its competitors. \label{DS4_boxplot}}
\end{figure}

In the combined analysis of DS3 and DS4, it is evident that distinct results arise depending on the model, despite their similar structures. In the case of DS3, funHDDC emerges as the superior procedure, exhibiting a significant performance advantage over the others. Conversely, in DS4, EHyClus takes the lead with a substantial margin compared to the other models. However, it is crucial to acknowledge that both strategies represent two highly effective approaches, with one outperforming the other in each respective case.

\section{Applications to real data} \label{real}

In this section, EHyClus for MFD is applied to two real datasets. The first is the widely studied Canadian Weather dataset, and the second is a dataset concerning air quality in Madrid.

\subsection{Canadian Weather data}

A popular real dataset in the FDA literature, included in \citet{ramsay} and in the \verb|fda| R-package, is the Canadian weather dataset. It contains the daily temperature and precipitation averaged over 1960 to 1994 at 35 different Canadian weather stations grouped into 4 different regions: Artic (3), Atlantic (15), Continental (12) and Pacific (5). The temperature and precipitation curves are represented in Fig.~\ref{cw_plot}, and the distribution of the 35 different stations in 4 regions is illustrated by Fig.~\ref{cw_map}. 

\begin{figure}[ht]
\centering
\includegraphics[width=\textwidth]{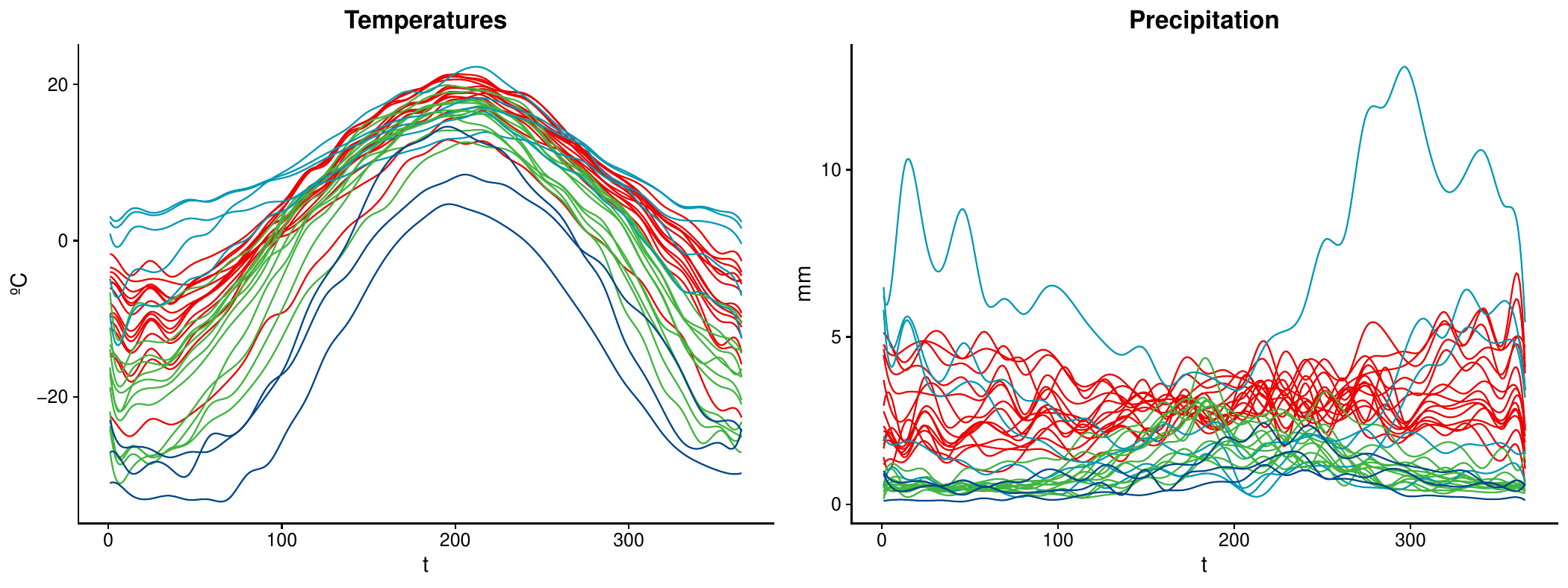}
\caption{Temperature and precipitation curves of 35 different Canadian weather stations, organized in four different climate zones. \label{cw_plot}}
\end{figure}

\begin{figure}[ht]
\centering
\includegraphics[width=0.55\textwidth]{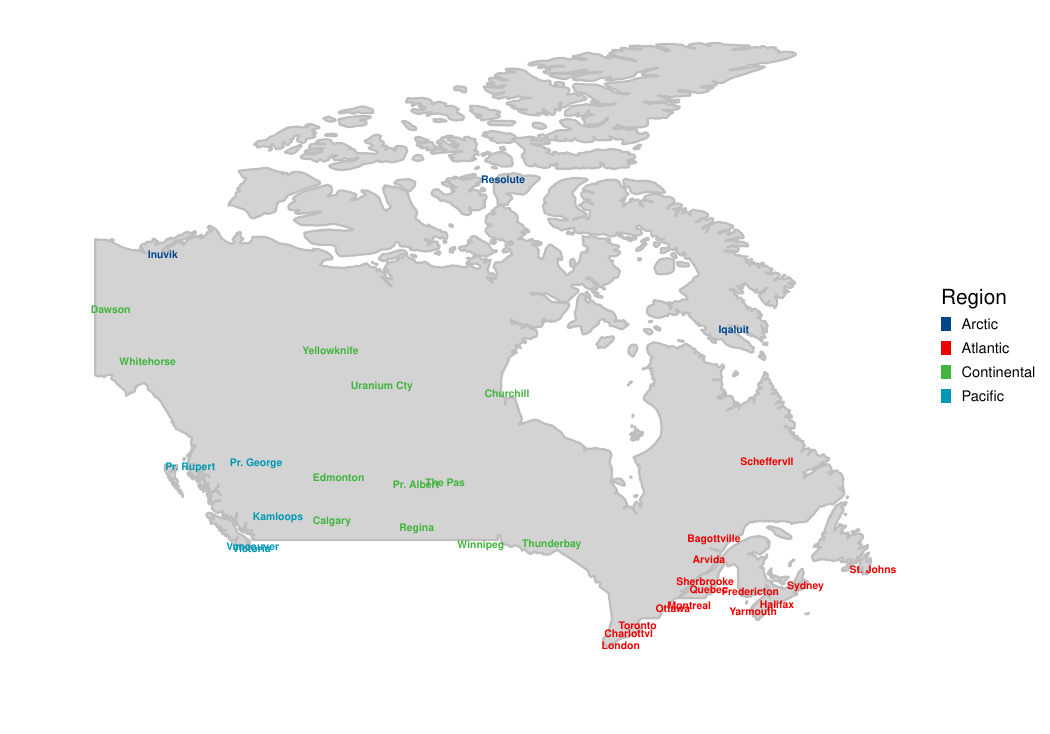}
\caption{Map of Canada with the name of the stations of four different regions represented with different colors. \label{cw_map}}
\end{figure}

In \citet{pulido2023}, EHyClus and some other cluster methodologies for functional data in one dimension were applied to cluster temperatures into four groups. The decision of generating four groups is based on the grouping in 4 regions given by the own dataset. This decision is also made in some other works, as \citet{jacques2014}, which provides a multivariate study in 4 clusters of temperature and precipitation. To do so, as the temperatures and precipitations are in different units, they normalize the data in order to properly work with it. In this paper, data normalization is unnecessary due to the utilization of \textbf{MEI} and \textbf{MHI}. These indices are applied to the curves, and consider the dimensionality of the data, respecting the units of the various dimensions when doing comparisons to the other curves. Consequently, the resultant values of \textbf{MEI} and \textbf{MHI} for a given curve are in the range between 0 and 1. As a result, the dataset derived from applying these indices to the original dataset is devoid of dissimilar scales, thereby obviating the need for data normalization.

First, we perform an analysis with 4 clusters and ground truth the division in regions as appear in Fig.~\ref{cw_map}. In this case, the best option in terms of the RI is considered for EHyClus, and all the methods for benchmarking are also considered. Table~\ref{cwcomp} presents the obtained results, being EHyClus with hierarchical clustering and Euclidean distance on the first derivatives the best approach between all the considered methods. The clusters obtained applying EHyClus with this combination appears in the left panel of Fig.~\ref{cw_mapEHyClus}.

\begin{table}

\centering
\begin{tabular}{lcccc}
\toprule
  & Purity & Fmeasure & RI & Time \\
\midrule
EHyClus & 0.7714 & 0.6768 & 0.7849 & 0.0002 \\ 
EHyClus-mean & 0.7429 & 0.5776 & 0.7714 & 0.0185\\
EHyClus-cov & 0.6857 & 0.5493 & 0.7160 & 0.0137 \\ 
Funclust & 0.4286& 0.4168 & 0.5345 & 0.0260 \\
funHDDC & 0.6571 & 0.4665  & 0.6924  & 0.9262 \\
FGRC & 0.6857 & 0.4892 & 0.6807 & 0.3491 \\
kmeans-d1 & 0.4286 & 0.2551 & 0.5681 & 0.1069\\
kmeans-d2 & 0.3530 &  0.3266 & 0.6626 & 0.4424 \\
gmfd-kmeans & 0.6286 & 0.4892 & 0.6807 & 0.6524 \\
\bottomrule
\end{tabular}
\caption{Purity, F-measure, Rand Index (RI) and execution time of Canadian Weather data for all the competitors models. \label{cwcomp}}
\end{table}

% \begin{figure}[ht]
% \centering
% \includegraphics[width=0.7\textwidth]{map_cw_ehyclus.pdf}
% \caption{Map of Canada with the names of the stations in four different colors. Each color represents a different cluster obtained with EHyClus. \label{cw_mapEHyClus}}
% \end{figure}

\begin{figure}[ht]
    \centering
    \begin{minipage}{0.49\textwidth}
        \centering
        \includegraphics[width=\textwidth]{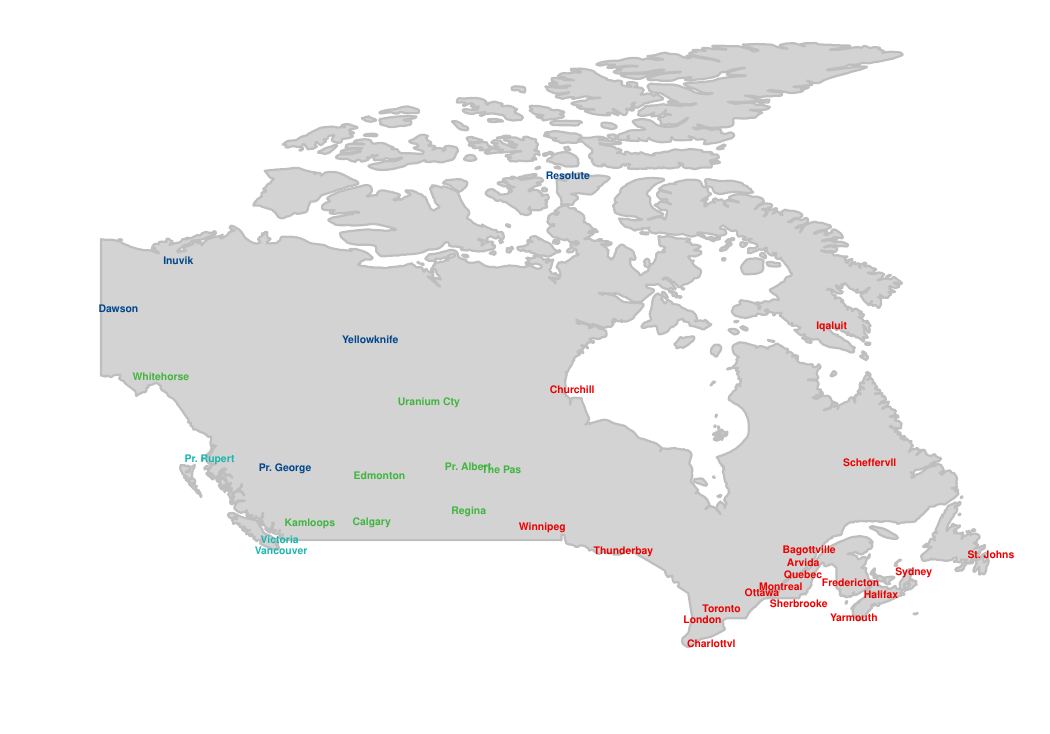}
        % \caption{Map of Canada with the names of the stations in four different colors. Each color represents a different cluster obtained with EHyClus.}
        % \label{cw_mapEHyClus}
    \end{minipage}
    \hfill
    \begin{minipage}{0.49\textwidth}
        \centering
        \includegraphics[width=\textwidth]{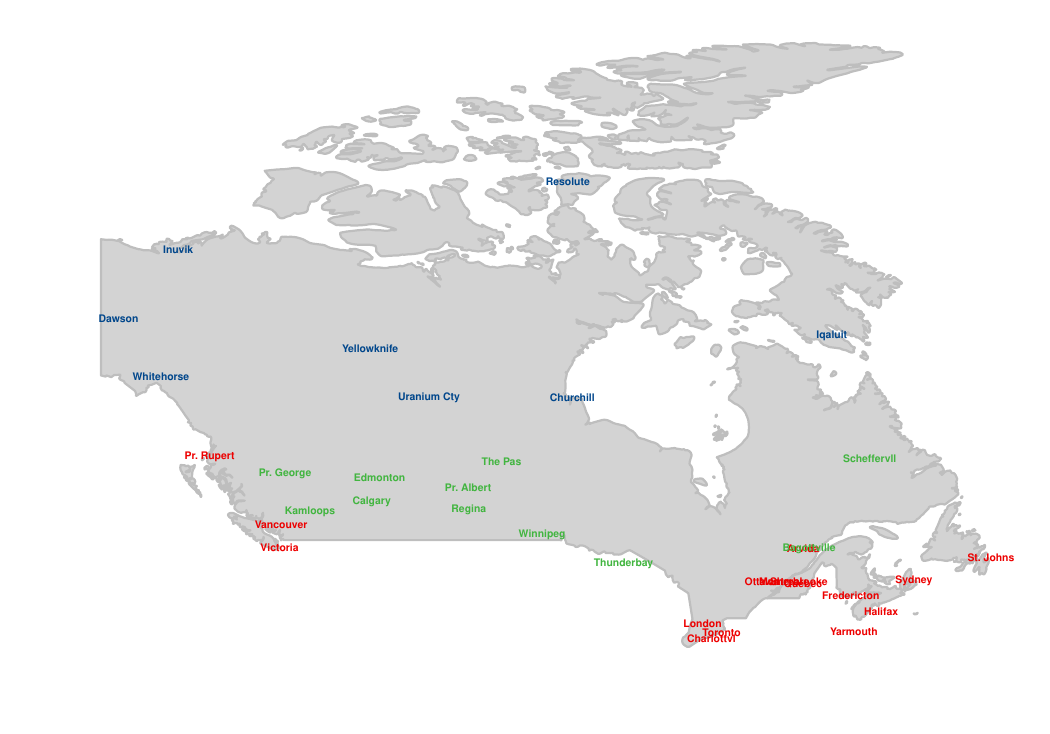}
        % \caption{Map of Canada with the names of the stations in three different colors. Each color represents a different cluster obtained with auto-EHyClus.}
        % \label{cw_mapEHyClus3}
    \end{minipage}
         \caption{Maps of Canada with the names of the stations in different colors. Left panel represents four clusters obtained with EHyClus. Right panel stands for three clusters obtained with auto-EHyClus.}
        \label{cw_mapEHyClus}
\end{figure}

The resulting groups share similar temperature and precipitation patterns, forming clusters with a clear geographical logic. Additionally, this map largely aligns with the regional distribution shown in Fig.~\ref{cw_map}, which has been used as the ground truth for Table~\ref{cwcomp}. However, differences arise at the boundaries of certain regions, such as Iqaluit in the Arctic region, Prince George and Kamloops in the Pacific region, and three stations in the Atlantic region: Churchill, Winnipeg, and Thunder Bay. The main limitation of this approach is the assumption that the regional classification in Fig.~\ref{cw_map} accurately reflects the behavior of temperature and precipitation. This may not always be the case.

As an alternative approach, we use the \verb|NbClust| R package, which considers 30 different indices to determine the optimal number of clusters. This method suggests three clusters, after which auto-EHyClus is applied. The resulting partition is obtained by applying k-means with Euclidean distance to \textbf{MEI} and \textbf{MHI} on the data, as well as the first and second derivatives. The final output is displayed in the right panel of Fig.~\ref{cw_mapEHyClus}. 

In this case, the 35 Canadian weather stations are grouped into three sets. The blue group is composed by northern Canada stations, characterized by subarctic or Arctic climates with long, harsh winters and short, cool summers. The red group includes stations along the Atlantic and Pacific coasts, which have maritime climates. The inclusion of central stations like Toronto and London can be attributed to their proximity to large lakes, which have a moderating effect on the climate. Finally, the green group includes stations in Central Canada that exhibit continental climates. Notably, Vancouver, Victoria, and Pr Rupert are included in the red group because of their location on the Pacific coast, whereas in the four-group classification they form a separate cluster.

This suggests that the four-group classification attempts to account for geographical details present in the ground truth, while auto-EHyClus focuses on defining climatic similarities based on temperature and precipitation data, aligning with the objectives of this analysis.

\subsection{Air quality data in Madrid}

This dataset examines air quality in Madrid, Spain’s capital, using open-access data sourced from the \href{https://datos.madrid.es/portal/site/egob}{Ayto. Madrid website}.  It provides hourly air quality measurements recorded throughout 2023, specifically tracking PM10 particles and nitrogen dioxide, with concentrations measured in $\mu g/m^3$. For this study, data from 13 monitoring stations in Madrid (Fig.~\ref{mad_curves}) were analyzed to investigate spatial patterns in air pollution. This approach allows for insights into the influence of urban design, traffic density, and green spaces on pollutant distribution. By classifying monitoring stations based on different pollutants, we can identify specific zones that are highly impacted, providing a foundation for targeted public health interventions and urban planning policies aimed at mitigating air quality issues in high-exposure areas.

\begin{figure}[ht]
\centering
\includegraphics[width=\textwidth]{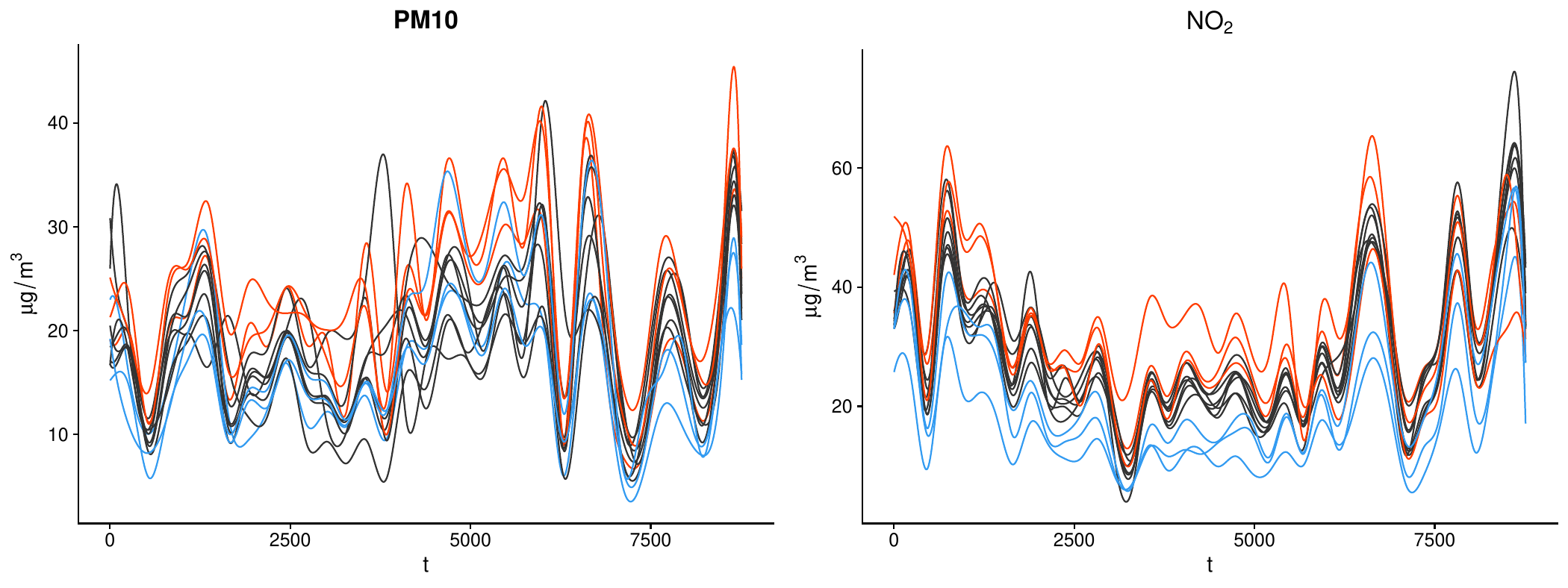}
\caption{PM10 and $\text{N0}_2$ curves of 13 different Madrid monitoring stations, grouped into three different clusters based on the concentrations of these two pollutants. The colors represent the three clusters obtained when applying EHyClus. \label{mad_curves}}
\end{figure}

To identify the optimal number of clusters, the \verb|NbClust| R package was used, which suggested three clusters as the best fit. The ``auto'' functionality of the EHyClus function from the \verb|ehymet| R package was then applied. The final classification, based on the application of k-means with Euclidean distance to the set composed of \textbf{MEI} on the first and second derivatives, and \textbf{MHI} on the data, as well as the first derivatives, resulted in the following three groups:

\begin{align*}
     \text{Cluster 1. } & \text{Escuelas de Aguirre, Urb. Embajada (Barajas), Plaza Elíptica.} \\
     \text{ Cluster 2. } & \text{C/ Farolillo, Moratalaz, Cuatro Caminos, Vallecas, Méndez Álvaro, }\\ 
     & \text{Paseo Castellana, Plaza Castilla.} \\
     \text{Cluster 3. } & \text{ Casa de Campo, Sanchinarro, Tres Olivos.}
\end{align*}
%which are represented with blue, green and red colors on the curves of PM10 and $\text{N0}_2$ in Fig.~\ref{mad_curves}.

\begin{figure}[ht]
\centering
\includegraphics[scale=0.41]{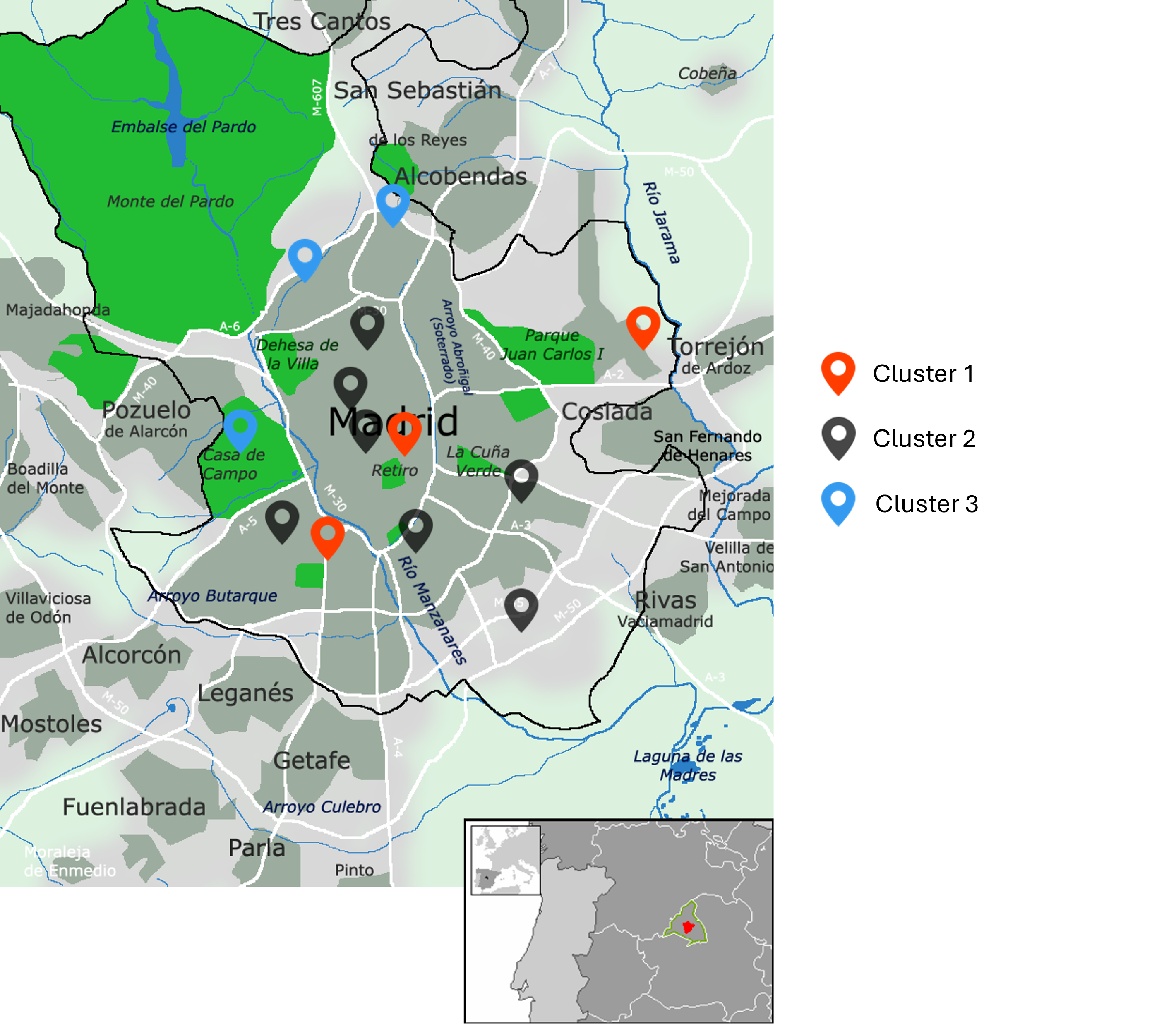}
\caption{Map of Madrid, Spain, showing weather stations grouped by pollution levels. Stations in the first cluster, represented in red, indicate areas with high pollution levels. The second cluster, shown in black, includes stations with moderate pollution, while the third cluster, shown in blue, represents stations with the lowest pollution levels. \label{mapaMadrid}}
\end{figure}

The three stations in the first group are likely impacted by high levels of traffic-related air pollution. Barajas (with its air traffic) and Plaza Elíptica (with road traffic) are both significant pollution hotspots, while Escuelas de Aguirre is influenced by its proximity to busy roads. The second cluster represents stations in moderately urbanized areas with varying levels of traffic and mixed commercial/residential zones. They are not as exposed as Plaza Elíptica or Barajas, but still experience considerable air pollution from both vehicle emissions and urban activities. Finally, the stations in the third group are in less urbanized areas, near green spaces (see Fig.~\ref{mapaMadrid}) or residential neighborhoods with less air pollution from traffic or industry. They generally show lower pollution levels compared to the more central or traffic-heavy zones. Thus, this classification, displayed in Fig.~\ref{mapaMadrid}, differentiates zones influenced by different sources of air pollution, based on factors such as traffic intensity, commercial activity, and proximity to green spaces.

%The first group, marked by high levels of pollution, corresponds to areas with intense traffic. The second group includes areas with moderate pollution, coming from urban traffic in mixed commercial and residential zones. Finally, the third group includes areas with generally better air quality, located far from traffic and pollution sources, often near large green spaces, as clearly seen in Fig.~\ref{mapaMadrid}.

\section{Conclusion} \label{conc}

The epigraph and hypograph indices, initially introduced by \citet{franc2011}, are fundamental tools for analyzing functional data in one dimension. However, extending these indices to the multivariate context is not straightforward, as it requires consideration of the interrelations among different dimensions. While previous attempts have extended these indices as a weighted average of the one-dimensional indices, we propose a novel multivariate formulation that goes beyond a combination of univariate measures.

In this study, we introduce the definitions of the univariate indices, the extension of the indices based on the weighted average of the univariate ones and our novel contribution, which takes into account the relations between different components. We also discuss the implications of adopting different definitions and their impact on the ordering of the indices. Theoretical properties of the proposed indices are also examined.
%Furthermore, we explore the relationship between the multivariate definition and the one-dimensional definitions in each dimension, highlighting that the multivariate indices are not a linear combination of those in one dimension. Finally, we discuss various theoretical properties associated with the multivariate indices.

The multivariate indices are then applied to the context of clustering using EHyClus, a methodology initially designed for univariate functional data and available in the \verb|ehymet| R package \citep{pulido2023}. By leveraging the proposed multivariate definition of the indices, we extend EHyClus to accommodate MFD. This option is also available in the package.  
We validate the efficacy of EHyClus by applying it to both simulated and real datasets, comparing its performance against existing approaches in the literature for clustering MFD. Our results show that EHyClus is highly competitive in terms of Purity, Rand Index (RI), and F-measure, while also demonstrating favorable execution times. Additionally, we introduce an automatic criterion for selecting one combination of data and indices, addressing the challenge of unknown ground truth in real-world applications, as exemplified by the Madrid air quality dataset.

Beyond clustering, the proposed multivariate indices offer potential for enhancing other index-based methodologies, such as the functional boxplot by \citet{martin2016} and the homogeneity test by \citet{franc2020}. 

%\clearpage
\bmhead{Acknowledgements} 
This research has been partially supported by Ministerio de Ciencia e Innovación, Gobierno de España, grant numbers  PTA2020-018802-I, PDC2022-133359, PID2022-137243OB-I00 and PID2022-137050NB-I00, TED2021-131264B-100 funded by MCIN/AEI/10.13039/501100011033 and European Union Next Generation EU/PRTR and ERDF A way of making Europe. This initiative has also
been partially carried out within the framework of Recovery,
Transformation and Resilience Plan funds,
financed by the European Union (Next Generation) through the grant
ANTICIPA and the ENIA 2022 Chairs for the creation of university-industry chairs in AI-AImpulsa: UC3M-Universia

\clearpage
\begin{appendices}

\section{Proofs}\label{secA1}

\subsection*{Proof of Theorem~\ref{rel} for the particular case of $p=3$.}

\begin{proof}

In this particular case, $$B^3_{1,2,3}=\sum_{i=1}^n \frac{\lambda(\{x_{i1}\leq x_1\} \cap \{x_{i2}\leq x_2\} \cap \{x_{i3}\leq x_3\})}{n\lambda(I)}.$$

Now, applying the rules of probability, 

\begin{equation*}
    \begin{split}
        B^3_{1,2,3}= &
        \sum_{i=1}^n \frac{\lambda(\{x_{i1}\leq x_1\} \cup \{x_{i2}\leq x_2\} \cup \{x_{i3}\leq x_3\})}{n \lambda(I)}\\ - & 
        \sum_{i=1}^n \frac{\lambda(x_{i1}\leq x_1)}{n \lambda(I)}-
        \sum_{i=1}^n \frac{\lambda(x_{i2}\leq x_2)}{n \lambda(I)}-
        \sum_{i=1}^n \frac{\lambda(x_{i3}\leq x_3)}{n \lambda(I)} \\+& 
        \sum_{i=1}^n \frac{\lambda(\{x_{i1}\leq x_1\}\cap \{x_{i2}\leq x_2\})}{n \lambda(I)}+
        \sum_{i=1}^n \frac{\lambda(\{x_{i1}\leq x_1\}\cap \{x_{i3}\leq x_3\})}{n \lambda(I)}\\+&
        \sum_{i=1}^n \frac{\lambda(\{x_{i2}\leq x_2\}\cap \{x_{i3}\leq x_3\})}{n \lambda(I)}.
    \end{split}
\end{equation*}

The expression above can be rewritten in terms of the definition of $B^3_{j_1,\ldots,j_r},$ with $\{j_1,\ldots,j_r\} \subseteq \{1,2,3\},$ and $r\leq 3,$ as follows:

\begin{equation*}
        B^3_{1,2,3}= \sum_{i=1}^n \frac{\lambda(\{x_{i1}\leq x_1\} \cup \{x_{i2}\leq x_2\} \cup \{x_{i3}\leq x_3\})} {n \lambda(I)}- B^3_{1} -B^3_{2}-B^3_{3}+B^3_{1,2}+B^3_{1,3}+B^3_{2,3}.
\end{equation*}

Taking into account that $\{x_{ij}\leq x_j\}^{\mathsf{c}} = \{x_{ij}> x_j\}, $

\begin{equation*}
        B^3_{1,2,3}= \sum_{i=1}^n \frac{\lambda(\{x_{i1}> x_1\}^{\mathsf{c}} \cup \{x_{i2}> x_2\}^{\mathsf{c}} \cup \{x_{i3}> x_3\}^{\mathsf{c}})}{n\lambda(I)}-  B^3_{1}-B^3_{2}-B^3_{3}+B^3_{1,2}+B^3_{1,3}+B^3_{2,3}.
    \label{b_p3}
\end{equation*}

And applying again the rules of probability,

\begin{equation*}
        B^3_{1,2,3}=  \sum_{i=1}^n \frac{\lambda(I)-\lambda(\{x_{i1}> x_1\} \cap \{x_{i2}> x_2\} \cap \{x_{i3}> x_3\})}{n\lambda(I)}-  B^3_{1}-B^3_{2}-B^3_{3}+B^3_{1,2}+B^3_{1,3}+B^3_{2,3}.
\end{equation*}

Now, taking into consideration that $\{x_{ij}\geq x_j\}$ can be written as the union of  two disjoint sets as follows: $$\{x_{ij}\geq x_j\} = \{x_{ij} > x_j\} \cup \{x_{ij} = x_j\}.$$

The following equality holds:
\begin{equation*}
\begin{split}
         \lambda(\{x_{i1}> x_1\} \cap \{x_{i2}> x_2\} \cap \{x_{i3}> x_3\}) = & \lambda(\{x_{i1}\geq x_1\} \cap \{x_{i2}\geq x_2\} \cap \{x_{i3}\geq x_3\})\\ - &
        \lambda(\{x_{i1}> x_1\} \cap \{x_{i2}= x_2\} \cap \{x_{i3}> x_3\}) \\ - &
        \lambda(\{x_{i1}= x_1\} \cap \{x_{i2}> x_2\} \cap \{x_{i3}> x_3\}) \\ - &
        \lambda(\{x_{i1}= x_1\} \cap \{x_{i2}= x_2\} \cap \{x_{i3}> x_3\}) \\ - &
        \lambda(\{x_{i1}> x_1\} \cap \{x_{i2}> x_2\} \cap \{x_{i3}= x_3\}) \\ - &
        \lambda(\{x_{i1}> x_1\} \cap \{x_{i2}= x_2\} \cap \{x_{i3}= x_3\}) \\ - &
        \lambda(\{x_{i1}= x_1\} \cap \{x_{i2}> x_2\} \cap \{x_{i3}= x_3\}) \\ - &
        \lambda(\{x_{i1}= x_1\} \cap \{x_{i2}= x_2\} \cap \{x_{i3}= x_3\}).
\end{split}
\end{equation*}

%Note that the curve $\mathbf{x}$ for which the indices are calculated is included in the sample. 
Consider $x_l$, with $1\leq l \leq n$, a curve in the sample. This implies that for $i=l$ $x_{ij}=x_{lj}$ and as a result $\lambda(\{x_{i1}= x_1\} \cap \{x_{i2}= x_2\} \cap \{x_{i3}= x_3\})=\lambda(I)$. Additionally, it holds that $\{x_{lj}> x_j\}=\emptyset$. This leads to the conclusion that all intersections other than the one containing all elements of the form ${x_{lj}=x_j}$ are empty sets.

Applying this to $B^3_{1,2,3}$, the following expression is obtained: 

\begin{equation}
\label{eq3}
    B^3_{1,2,3}= -B^3_{1}-B^3_{2}-B^3_{3}+B^3_{1,2}+B^3_{1,3}+B^3_{2,3} + 1 - A^3_{1,2,3} +\frac{1}{n}+R_3,
\end{equation}
where 

\begin{align*}
        R_3= &
        \sum_{\substack{i=1 \\ i \neq l}}^n \frac{\lambda(\{x_{i1}> x_1\} \cap \{x_{i2}= x_2\} \cap \{x_{i3}> x_3\})}{n\lambda(I)}\\ + &
        \sum_{\substack{i=1 \\ i \neq l}}^n \frac{\lambda(\{x_{i1}= x_1\} \cap \{x_{i2}> x_2\} \cap \{x_{i3}> x_3\})}{n\lambda(I)}\\ + &
        \sum_{\substack{i=1 \\ i \neq l}}^n \frac{\lambda(\{x_{i1}= x_1\} \cap \{x_{i2}= x_2\} \cap \{x_{i3}> x_3\})}{n\lambda(I)}\\ + &
        \sum_{\substack{i=1 \\ i \neq l}}^n \frac{\lambda(\{x_{i1}> x_1\} \cap \{x_{i2}> x_2\} \cap \{x_{i3}= x_3\})}{n\lambda(I)}\\ + &
        \sum_{\substack{i=1 \\ i \neq l}}^n \frac{\lambda(\{x_{i1}> x_1\} \cap \{x_{i2}= x_2\} \cap \{x_{i3}= x_3\})}{n\lambda(I)}\\ + &
        \sum_{\substack{i=1 \\ i \neq l}}^n \frac{\lambda(\{x_{i1}= x_1\} \cap \{x_{i2}> x_2\} \cap \{x_{i3}= x_3\})}{n\lambda(I)}\\ + &
        \sum_{\substack{i=1 \\ i \neq l}}^n \frac{\lambda(\{x_{i1}= x_1\} \cap \{x_{i2}= x_2\} \cap \{x_{i3}= x_3\})}{n\lambda(I)}.
\end{align*}
Thus, applying Equations~\eqref{mei_A} and~\eqref{mhi_B} to Equation~\eqref{eq3}:
\begin{equation*}
    \begin{split}
         \mathbf{MHI}_n(\mathbf{x_k})-\mathbf{MEI}_n(\mathbf{x_k})=& \mathbf{MHI}^3_{n,1,2}(\mathbf{x_k})+\mathbf{MHI}^3_{n,1,3}(\mathbf{x_k})+\mathbf{MHI}^3_{n,2,3}(\mathbf{x_k}) \\ - & 
         \mathbf{MHI}^3_{n,1}(\mathbf{x_k})-\mathbf{MHI}^3_{n,2}(\mathbf{x_k})-\mathbf{MHI}^3_{n,3}(\mathbf{x_k})+\frac{1}{n}+R_3.
    \end{split}
\end{equation*}
\end{proof}

\subsection*{Proof of Theorem~\ref{rel} for the general case.}

\begin{proof} 
Applying the rules of probability and the definition of $B_{j_1,\ldots,j_r}$ given by Equation~\eqref{bp}:
\begin{equation*}
    \begin{split}
        B^p_{1,\ldots,p} & = 
        (-1)^{p+1} \sum_{i=1}^n \frac{\lambda(\bigcup_{j=1}^p \{x_{ij} \leq x_j\})}{n\lambda(I)}+ \sum_{r=1}^{p-1}\sum_{1\leq j_1<\ldots<j_r\leq p}^p (-1)^{r+p+1} B_{n,j_1,\ldots,j_r} \\ & = 
        (-1)^{p+1} \sum_{i=1}^n \frac{\lambda(\bigcup_{j=1}^p \{x_{ij} > x_j\}^{\mathsf{c}})}{n\lambda(I)}+\sum_{r=1}^{p-1}\sum_{1\leq j_1<\ldots<j_r\leq p}^p (-1)^{r+p+1} B_{n,j_1,\ldots,j_r}\\ & = 
        (-1)^{p+1} \sum_{i=1}^n \frac{\lambda(I)-\lambda(\bigcap_{j=1}^p \{x_{ij} > x_j\})}{n\lambda(I)}+\sum_{r=1}^{p-1}\sum_{1\leq j_1<\ldots<j_r\leq p}^p (-1)^{r+p+1} B_{n,j_1,\ldots,j_r}\\ & = 
        \sum_{r=1}^{p-1}\sum_{1\leq j_1<\ldots<j_r\leq p}^p (-1)^{r+p+1} B^p_{j_1,\ldots,j_r}+(-1)^{p+1}+(-1)^p A^p_{1,\ldots p}\\  & +
        (-1)^{p+1}\frac{1}{n}+(-1)^{p+1}R_p,
    \end{split}
\end{equation*}
where $R_p = \sum_{k=1}^{2^p-1} \sum_{\substack{i=1 \\ i \neq j}}^n \frac{C}{n\lambda(I)},$ with $C \in \mathcal{C}_p$, where $\mathcal{C}_p$ is the set of the Lebesgue measure of all the possible intersections of $p$ elements of the type $\{x_{ij}> x_j\}$ or $\{x_{ij}= x_j\}$, $j=1,\ldots,p$. It is important to note that the set $\mathcal{C}_p$ is composed by $2^p$ elements. Nevertheless, the above summation is taken up to $2^p-1$ since the intersection that contains all the elements of type $\{x_{ij}> x_j\}$ is included in the disaggregation of $B^p_{1,\ldots,p}$.

Finally, the following relation is obtained for $\textbf{x}_k, \ 1 \leq k \leq n,$ a curve in the sample:
\begin{align*}
     \mathbf{MHI}_n(\mathbf{x}_k)+(-1)^{p}\mathbf{MEI}_n(\mathbf{x_k}) & = \sum_{r=1}^{p-1}\sum_{1\leq j_1 \ldots j_r\leq p}^p (-1)^{r+p+1} \mathbf{MHI}^p_{n,j_1,\ldots,j_r}(\mathbf{x_k}) \\ & +(-1)^{p+1}\frac{1}{n}+(-1)^{p+1}R_p.
\end{align*}

\end{proof}

\subsection*{Proof of Proposition~\ref{p1}.}

\begin{proof}
The proof for the epigraph is given here. The one for the hypograph index can be obtained in the same way.
\begin{itemize}
    \item[a.]  By definition,
    $$\mathbf{EI}(\mathbf{x}) =  1-P(\bigcap_{k=1}^p\{X_k(t)\geq x_k(t),t\in \mathcal{I}\}).$$   
    Thus, $X_k(t)\geq x_k(t)$, if and only if, $A_{k}(t)X_{k}(t)+b_k(t)\geq A_{k}(t) x_k(t)+b_k(t)$ and therefore, $$\mathbf{EI}(\mathbf{T}(\mathbf{x}))=\mathbf{EI}(\mathbf{x}).$$
    \item[b.] Given g is a one-to-one transformation of the interval $\mathcal{I}$ to $\mathcal{I}$, $X_{k}\geq x_k$, if and only if, $X_{k}(g)\geq x_k(g)$ $(t \leftrightarrow g)$. Therefore, $$\mathbf{EI}(\mathbf{x}(g))=\mathbf{EI}(\mathbf{x}).$$
\end{itemize}
\end{proof}

\subsection*{Proof of Proposition~\ref{p1_2}.}
\begin{proof}
The proof for the epigraph is given here. The one for the hypograph index can be obtained in the same way.
\begin{itemize}
    \item[a.]  By definition, $$\mathbf{MEI}(\mathbf{x})=  1-\frac{E(\lambda(\bigcap_{k=1}^p\{t \in \mathcal{I} : X_k(t)\geq x_k(t)\}))}{n\lambda(\mathcal{I})}$$ Thus, $X_{k}(t)\geq x_k(t)$, if and only if, $A_{k}(t)X_{k}(t)+b_k(t)\geq A_{k}(t) x_k(t)+b_k(t)$, and therefore, $$\mathbf{MEI}(\mathbf{T}(\mathbf{x}))=MEI(\mathbf{x}).$$
    
    \item[b.]  Given g is a one-to-one transformation of the interval $\mathcal{I}$ to $\mathcal{I}$, $(t \leftrightarrow g)$, $$\lambda(\bigcap_{k=1}^p\{t \in \mathcal{I} : X_k(t)\geq x_k(t)\}) = \lambda(\bigcap_{k=1}^p\{t \in \mathcal{I} : X_k(g(t))\geq x_k(g(t))\}),$$
    and thus, $$\mathbf{MEI}(\mathbf{x}(g))=\mathbf{MEI}(\mathbf{x}).$$
\end{itemize}
\end{proof}

\subsection*{Proof of Proposition~\ref{p2}.}
\begin{proof}
   Consider $D_M = \{x:\min_{1\leq k \leq p} \norm{x_k}_{\infty} \geq M\}$, and let prove that $$\underset{x \in D_M}{\sup} \max \{\mathbf{EI}(\mathbf{x},P_{\mathbf{X}}), 1-\mathbf{HI}(\mathbf{x},P_{\mathbf{X}})\}  \to 1, \text{when} \ M \to \infty.$$ 
   
    By definition, we can write the indices as follows
   $$\textbf{EI}(\textbf{x})=1-P\left(\bigcap_{k=1}^p\{X_k(t)\geq x_k(t),\text{ for all }t\in \mathcal{I}\}\right)=1-P(\cap_{k=1}^p A_k),$$
   and
   $$1-\textbf{HI}(\textbf{x})=1-P\left(\bigcap_{k=1}^p\{X_k(t)\leq x_k(t),\text{ for all }t\in \mathcal{I}\}\right)=1-P(\cap_{k=1}^p B_k),$$
   where $A_k=\{X_k(t)\geq x_k(t),\text{ for all } t\in \mathcal{I}\}$ and $B_k=\{X_k(t)\leq x_k(t), \text{ for all }t\in \mathcal{I}\}.$

   Now,
       \begin{align*}
        \max \{\mathbf{EI}(\mathbf{x},P_{\mathbf{X}}), 1-\mathbf{HI}(\mathbf{x},P_{\mathbf{X}})\} = & \max \{1-P(\cap_{k=1}^p A_k), 1-P(\cap_{k=1}^p B_k)\}  \\ = & 
        1-\min\{P(\cap_{k=1}^p A_k), P(\cap_{k=1}^p B_k)\}.
    \end{align*}

    Thus, the proof of this proposition is equivalent to prove that 
    $$\sup_{\textbf{x}\in D_M} \min\{P(\cap_{k=1}^p A_k), P(\cap_{k=1}^p B_k)\} \to 0, \text{when} \ M \to \infty.$$

    The following inequality holds:
    $$\min\{P(\cap_{k=1}^p A_k), P(\cap_{k=1}^p B_k)\} \leq \max_k \min\{P(A_k), P( B_k)\},$$
    and by Propositions 1 and 5 in \citet{lop2011}, we have that, for all $k \in \{1,\ldots,p\}$,
    $$\sup_{\norm{x_k}_{\infty}\geq M} \min\{P(A_k), P( B_k)\} \to 0, \text{when} \ M \to \infty.$$

    Then,
    \begin{align*}
        \sup_{\textbf{x}\in D_M} \min\{P(\cap_{k=1}^p A_k), P(\cap_{k=1}^p B_k)\} \leq &\sup_{\textbf{x}\in D_M} \max_k \min \{P(A_k), P( B_k)\}  \\ \leq &
    \sup_{\textbf{x}\in D_M} \min \{P(A_k), P( B_k)\} \to 0, \text{when} \ M \to \infty
    \end{align*}
    
   Let prove now that $$\underset{x \in D_M}{\sup} \max \{\mathbf{EI}_n(\mathbf{x}), 1-\mathbf{HI}_n(\mathbf{x})\}  \overset{a.s}{\to} 1, \ \text{ when } M \ \to \infty.$$

    In this case we consider $A_{i,k}= \{x_{i,k}(t)\geq x_k(t), \text{ for all } t\in \mathcal{I}\},$ and \\ $B_{i,k}= \{x_{i,k}(t)\leq x_k(t), \text{ for all } t\in \mathcal{I}\},$ to have that
    $$\mathbf{EI}_n(\textbf{x})=1-\frac{1}{n}\sum_{i=1}^n I(\cap_{k=1}^p A_{i,k}), $$ and
    $$1-\mathbf{HI}_n(\textbf{x})=1-\frac{1}{n}\sum_{i=1}^n I(\cap_{k=1}^p B_{i,k}).$$

    Now,
       \begin{align*}
        \max \{\mathbf{EI}_n(\mathbf{x}), 1-\mathbf{HI}_n(\mathbf{x})\} = & \max \{1-\frac{1}{n}\sum_{i=1}^n I(\cap_{k=1}^p A_k), 1-\frac{1}{n}\sum_{i=1}^n I(\cap_{k=1}^p B_k)\}  \\ = & 
        1-\min\{\frac{1}{n}\sum_{i=1}^n I(\cap_{k=1}^p A_k), \frac{1}{n}\sum_{i=1}^n I(\cap_{k=1}^p B_k)\}.
    \end{align*}

   Again, to prove that $$\underset{\textbf{x} \in D_M}{\sup} \max \{\mathbf{EI}_n(\mathbf{x}), 1-\mathbf{HI}_n(\mathbf{x})\}  \overset{a.s}{\to} 1, \ \text{ when } M \ \to \infty,$$ is equivalent to prove $$\underset{\textbf{x} \in D_M}{\sup} \min\{\frac{1}{n}\sum_{i=1}^n I(\cap_{k=1}^p A_k), \frac{1}{n}\sum_{i=1}^n I(\cap_{k=1}^p B_k)\} \overset{a.s}{\to} 0, \ \text{ when } M \ \to \infty.$$

   By Proposition 5 in \citet{lop2011}, we also have that, for all $k \in \{1,\ldots,p\}$, $$\sup_{\norm{x_k}_{\infty}\geq M} \min\{\frac{1}{n}\sum_{i=1}^n I(A_{i,k}), \frac{1}{n}\sum_{i=1}^n I(B_{i,k})\} \to 0, \text{when} \ M \to \infty.$$

   Thus, 
    \begin{align*}
        & \sup_{\textbf{x}\in D_M} \min\{\frac{1}{n}\sum_{i=1}^n I(\cap_{k=1}^p A_{i,k}), \frac{1}{n}\sum_{i=1}^n I(\cap_{k=1}^p B_{i,k})\}  \\ \leq & 
        \sup_{\textbf{x}\in D_M} \max_k \min \{\frac{1}{n}\sum_{i=1}^n I(A_{i,k}), \frac{1}{n}\sum_{i=1}^n I(B_{i,k})\} \\ \leq  &
    \sup_{\textbf{x}\in D_M} \min \{\frac{1}{n}\sum_{i=1}^n I(A_{i,k}), \frac{1}{n}\sum_{i=1}^n I(B_{i,k})\} \to 0, \text{when} \ M \to \infty.
    \end{align*}
\end{proof}

\clearpage

\section{Tables with notation} \label{secA2}

\subsection*{Table displaying the combinations of data and indexes considered in this work.}
\begin{table}[htbp]
\centering
\begin{tabular}{@{}p{6cm}p{9cm}@{}}
\toprule
 Notation & Description\\
\midrule
\_.MEIMHI = (\textbf{MEI}, \textbf{MHI})  & The modified epigraph and the hypograph index on the original curves.  \\
d.MEIMHI = (\textbf{dMEI}, \textbf{dMHI})  & The modified epigraph and the hypograph index on the first derivatives.  \\
d2.MEIMHI = (\textbf{d2MEI}, \textbf{d2MHI})  & The modified epigraph and the hypograph index on the second derivatives.  \\
\_d.MEIMHI = (\textbf{MEI}, \textbf{MHI}, \textbf{dMEI}, \textbf{dMHI})  & The modified epigraph and the hypograph index on the original curves and on the first derivatives. \\
\_d2.MEIMHI = (\textbf{MEI}, \textbf{MHI}, \textbf{d2MEI}, \textbf{d2MHI})  & The modified epigraph and the hypograph index on the original curves and on the second derivatives. \\
dd2.MEIMHI = (\textbf{dMEI}, \textbf{dMHI}, \textbf{d2MEI}, \textbf{d2MHI})  & The modified epigraph and the hypograph index on the first and on the second derivatives. \\
\_dd2.MEIMHI = (\textbf{MEI}, \textbf{MHI}, \textbf{dMEI}, \textbf{dMHI}, \textbf{d2MEI}, \textbf{d2MHI})  & The modified epigraph and the hypograph index on the original curves, first and second derivatives. \\
\_d.MEI = (\textbf{MEI}, \textbf{dMEI}) & The modified epigraph index on the original curves and first derivatives. \\
\_d2.MEI = (\textbf{MEI},\textbf{d2MEI}) & The modified epigraph index on the original curves and on the second derivatives. \\
dd2.MEI = (\textbf{dMEI}, \textbf{d2MEI}) & The modified epigraph index on the first and on the second derivatives. \\
\_dd2.MEI = (\textbf{MEI}, \textbf{dMEI}, \textbf{d2MEI}) & The modified epigraph index on the original curves, first and second derivatives. \\
\_d.MHI = (\textbf{MHI}, \textbf{dMHI}) & The modified hypograph index on the original curves and on the first derivatives. \\
\_d2.MHI = (\textbf{MHI}, \textbf{d2MHI}) & The modified hypograph index on the original curves and on the second derivatives. \\
dd2.MHI = (\textbf{dMHI}, \textbf{d2MHI}) & The modified hypograph index on the first and son the econd derivatives. \\
\_dd2.MHI = (\textbf{MHI}, \textbf{dMHI}, \textbf{d2MHI}) & The modified hypograph index on the original curves, first and second derivatives. \\
\bottomrule
\end{tabular}
\caption{Notation and description of the combinations of data and indices.\label{data}}
\end{table}

\newpage
\subsection*{Table displaying the clustering method applied to the resulting multivariate dataset}
\begin{table}[htbp]

\centering
\begin{tabular}{@{}p{4cm}p{9cm}@{}}
\toprule
 Notation & Description\\
\midrule
single.(b).(c)  & Hierarchical clustering with single linkage and Euclidean distance.  \\
complete.(b).(c) & Hierarchical clustering with complete linkage and Euclidean distance.   \\
average.(b).(c) & Hierarchical clustering with average linkage and Euclidean distance.    \\
centroid.(b).(c) & Hierarchical clustering with centroid linkage and Euclidean distance.   \\
ward.D2.(b).(c) & Hierarchical clustering with Ward method and Euclidean distance.   \\
kmeans.(b).(c)-euclidean  & k-means clustering with Euclidean distance. \\
kmeans.(b).(c)-mahalanobis  & k-means clustering with Mahalanobis distance. \\
kkmeans.(b).(c)-gaussian  & kernel k-means clustering with a Gaussian kernel. \\
kkmeans.(b).(c)-polynomial  & kernel k-means clustering with a polynomial kernel. \\
spc.(b).(c) & spectral clustering. \\
svc.(b).(c)-kmeans  & support vector clustering with k-means initialization. \\
svc.(b).(c)-kkmeans  & support vector clustering with kernel k-means initialization. \\
\bottomrule
\end{tabular}
\caption{Notation and description of the clustering method applied to the dataset obtained from the combination of data and indices given by (b).(c).\label{comb}}
\end{table}




\end{appendices}


\clearpage
\addcontentsline{toc}{chapter}{Bibliography}

\printbibliography
%\bibliography{sn-article}% common bib file

\end{document}